\documentclass[english,twocolumn,transaction]{IEEEtran}
\usepackage[T1]{fontenc}
\usepackage{float}
\usepackage{amsmath}
\usepackage{amsthm}
\usepackage{amssymb}
\usepackage{stackrel}
\usepackage{graphicx}
\usepackage{setspace}
\usepackage{url}
\usepackage{algorithm}
\usepackage{algorithmic}
\usepackage{booktabs}
\usepackage{stfloats}
\usepackage{cite}
\usepackage{float}

\usepackage{subfigure}
\makeatletter

\floatstyle{ruled}
\newfloat{algorithm}{tbp}{loa}
\providecommand{\algorithmname}{Algorithm}
\floatname{algorithm}{\protect\algorithmname}

\theoremstyle{plain}

\theoremstyle{definition}

\theoremstyle{plain}

\theoremstyle{plain}

\newcommand{\RNum}[1]{\uppercase\expandafter{\romannumeral #1\relax}}

\IEEEoverridecommandlockouts
\usepackage{ifpdf}
\usepackage{cite}
\ifCLASSOPTIONcompsoc
\usepackage[caption=false,font=normalsize,labelfont=sf,textfont=sf]{subfig}
\else
\usepackage[caption=false,font=footnotesize]{subfig}
\fi
\usepackage{amsmath}
\@ifundefined{showcaptionsetup}{}{%
 \PassOptionsToPackage{caption=false}{subfig}}
\usepackage{subfig}
\makeatother

\usepackage{colortbl}

\usepackage{babel}

\renewcommand\figurename{Fig.}

\begin{document}

\title{Energy-Efficient Resource Allocation in Massive MIMO-NOMA Networks with Wireless Power Transfer: A Distributed ADMM Approach} \vspace{-0.4em}
\author{Zhongyu Wang,~\emph{Student~Member,~IEEE},~Zhipeng~Lin,~\emph{Member,~IEEE},~Tiejun~Lv,~\emph{Senior Member,~IEEE},~and~Wei~Ni,~\emph{Senior~Member,~IEEE} \vspace{-2.5em}

\thanks{Manuscript received January 18, 2020; revised June 11, 2020; accepted March 19, 2021.
\emph{(Corresponding author: Tiejun Lv.)}

Z. Wang, Z. Lin and T. Lv are with the School of Information and Communication Engineering, Beijing University of Posts and Telecommunications (BUPT), Beijing 100876, China (e-mail: \{zhongyuwang, linlzp, lvtiejun\}@bupt.edu.cn).

W. Ni is with the Commonwealth Scientific and Industrial Research Organisation (CSIRO), Sydney 2122, Australia (e-mail: Wei.Ni@data61.csiro.au).
}}

\maketitle
\begin{abstract}
In multicell massive multiple-input multiple-output (MIMO) non-orthogonal multiple access (NOMA) networks, base stations (BSs) with multiple antennas deliver their radio frequency energy in the downlink, and Internet-of-Things (IoT) devices use their harvested energy to support uplink data transmission. This paper investigates the energy efficiency (EE) problem for multicell massive MIMO NOMA networks with wireless power transfer (WPT). To maximize the EE of the network, we propose a novel joint power, time, antenna selection, and subcarrier resource allocation scheme,
which can properly allocate the time for energy harvesting and data transmission.
Both perfect and imperfect channel state information (CSI) are considered, and their corresponding EE performance is analyzed.
Under quality-of-service (QoS) requirements, an EE maximization problem is
formulated, which is non-trivial due to non-convexity.
We first adopt nonlinear fraction programming methods to convert the problem to be convex, and then, develop a distributed alternating direction method of multipliers (ADMM)-based approach to solve the problem.
Simulation results demonstrate
that compared to alternative methods, the proposed algorithm can converge quickly within fewer
iterations, and can achieve better EE performance.
\end{abstract}

\vspace{-1em}
\section{Introduction }
\renewcommand\figurename{Fig.}
\subsection{Background and Motivation}

Explosive developments of the Internet-of-Things (IoT) stimulate a huge demand for massive connectivity capability, which enables a large number of IoT devices to access wireless networks \cite{8930983}.
The accessed IoT devices would rapidly drain their batteries if they are used in high power-consumption applications \cite{8534441}.
Therefore, it is important to ensure that massive energy-constrained IoT devices can communicate with each other in a spectrally- and energy-efficient manner while meeting quality of service (QoS) requirements \cite{8214104}.
Wireless power transfer (WPT) uses uninterrupted
power supply methods to remotely charge IoT devices with small rechargeable batteries \cite{7010878}.
Such technique enables battery-powered IoT devices to extract the radio frequency (RF) signal energy to prolong their lifecycles \cite{8485639}.
As a result, continuous power supply requirements of IoT devices can be met \cite{8444982}.

Massive MIMO has been regarded as an important technology to promote the deployment of the fifth-generation and beyond (B5G) networks \cite{8485639}.
Utilizing beams for spatial multiplexing, massive MIMO can effectively concentrate received signal power with very narrow beams and achieve very dense deployment for IoT networks.
Massive MIMO is also an enable technique for for WPT systems, because it can provide system-level energy efficiency (EE)\cite{8629017, 6457363}.
In \cite{7062017}, user scheduling, power allocation, and rate adaptation scheme were jointly optimized to improve the EE of orthogonal frequency division multiple access (OFDMA) networks.
Nevertheless, aforementioned transmission schemes \cite{8629017,6457363,7062017} are carried out in an orthogonal multiple access (OMA) manner. OMA is restrictive in terms of spectral efficiency (SE) for massive IoT devices accessed networks, because each frequency subchannel can only be used by a single user at a time.

To effectively use communication spectrums and improve system connectivity, non-orthogonal multiple access (NOMA) has been proposed, which can guarantee multiple devices access the same subchannel simultaneously in the power domain \cite{8390925}.
\cite{add20200411} proposed a mobile user association scheme for small-cell networks.
The scheme used NOMA and successive interference cancellation (SIC) to achieve the system utility maximization and total transmit power minimization.
NOMA can also enlarge the capacity region of OMA via the power multiplexing and SIC at transceivers \cite{1291726}. Combining NOMA with MIMO (referred to as ``MIMO-NOMA'') is able to significantly enhance the SE, channel capacity and EE \cite{8638930,8301007}.
Under the constraint of QoS, an EE maximization problem was studied in an uplink NOMA-based millimeter-wave massive MIMO system \cite{8603758}.
By installing large scale antennas at a base station (BS), the transmit power can be concentrated in some areas, consequently improving the throughput and channel rate of massive MIMO networks \cite{7031971}.
Nevertheless, employing massive antenna arrays would dramatically increase energy consumption caused by of IoT networks due to the deployment of massive antennas \cite{8474292}.
In this sense, it is important to find an effective method, which can properly select the antenna number to optimize the EE of the considered multicell WPT-enabled massive MIMO-NOMA networks.

The architecture of WPT-based massive MIMO-NOMA networks
is envisioned to be a prospective approach to enhance the SE and EE of the IoT networks.
As the B5G wireless communications have various radio resources (e.g., time, frequency,
spatial, power and code domains, etc.) \cite{8058662},
it is critical for massive MIMO-NOMA IoT to achieve an energy-efficient resource allocation with incessantly energy supply.
It is also vital to find efficient resource allocation strategies of radio resource for practical resource-constrained massive MIMO-NOMA IoT environments with imperfect channel state information (CSI). This is because perfect CSI is hard to be acquired. The accuracy of the obtained CSI is dependent on  channel feedback and quantization errors \cite{7009979}.
As a result, we need to exploit NOMA for WPT to implement the massive connectivity for IoT networks, and properly allocate resources in the distributed massive MIMO-NOMA networks with multiple antenna selection (AS).
\vspace{-1em}
\subsection{Contribution}

This paper considers the multicell massive MIMO-NOMA network, where the WPT-powered IoT devices
transmit data after they harvest sufficient energy from the BSs.
We divide the time block into two parts, i.e., the WPT time and the wireless information transmission (WIT) time.
The AS technique, which can improve the throughput and EE performance
of the system, is applied to avoid excessive cost (such as hardware cost, device power and software computation).
To maximize the EE of the massive MIMO-NOMA network, we propose a distributed alternating-direction-method-of-multipliers (ADMM)-based resource allocation algorithm. The main contributions are summarized as follows.

\begin{itemize}

 \item
 We propose a WPT-based distributed massive MIMO-NOMA model, in which multi-antenna BSs and massive rechargeable IoT devices are deployed, to meet the massive connectivity requirement and achieve the high EE.
 In the power domain, we employ NOMA technique to achieve SIC and improve the EE. In the time domain, we design a two-time-slot division protocol, which can effectively divide the WPT and WIT time.

 \item
 We derive the signal-to-interference-plus-noise ratio (SINR) and EE with the multiple AS under the perfect and imperfect CSI.
 A new scheme, which jointly optimizes the power control, time division, AS and subcarrier assignment, is proposed, and its corresponding EE optimization problem is formulated for the distributed multicell massive MIMO-NOMA networks.

 \item
 We apply nonlinear optimization to convert the original non-convex fractional
 programming problem to be convex.
 A novel distributed ADMM-based energy-efficient resource allocation algorithm is also proposed to divide the global consensus problem into several sequential optimization phases. In this way, the optimal solution of the problem can be obtained.

 \item
 Extensive simulations under different parameters show substantial EE improvements of the proposed distributed ADMM-based resource allocation algorithm in the massive MIMO-NOMA networks.
 We also demonstrate that the proposed distributed ADMM-based algorithm converges well and
 outperforms the existing benchmarks in terms of overall EE.
\end{itemize}
\vspace{-1em}
\subsection{Organization}
The rest of this paper is organized as follows. Section II is literature review.
The system model and the communication process are introduced in Section III. Section IV formulates the EE optimization problem. In Section V, a distributed resource allocation algorithm via ADMM is proposed to address the EE maximization. Section VI shows the simulation results. Finally, the conclusions and future works are drawn in Section VII.

\emph{Notation:} $a$, $\mathbf{a}$, $\mathcal{A}$ and $\mathbf{A}$ stand for a scalar, a column vector, a set and a matrix, respectively; $\mathbf{I}$ represents the identity matrix; $\mathbf{A}^H$ represents the Hermitian transpose of $\mathbf{A}$; $\| \cdot \|$ and $| \cdot |$ denote the L-norm and norm of a vector; $E\left[ \cdot \right]$ denotes the expectation; $\mathcal{C} \mathcal{N}(a,b)$ denotes the distribution of a circularly symmetric complex Gaussian (CSCG) vector with mean $a$ and variance $b$; $\mathcal{O}(\cdot)$ denotes the computational complexity.

\section{Related Work}

In recent years, WPT has widely attracted attention in academia and industry, as it can extend the lifetime of energy-limited IoT devices \cite{8558585}.
\cite{7942090} proposed a simultaneous wireless information and power transfer scheme to maximize the minimum achievable rate of all IoT devices by using power-splitting techniques in large-scale antenna array systems.
\cite{7009979} maximized the throughput of a WPT-powered massive MIMO system.
\cite{6884177} studied the downlink and uplink EE tradeoff in a time-division duplexing (TDD) OFDMA system, which has a single access point and multiple WPT-enabled IoT devices.
However, the optimization schemes in \cite{8558585,7942090,7009979,6884177} only considered OMA transmissions, and their achieved channel capacity and SE were restrictive.
Due to dense deployment of IoT devices and the massive connectivity requirement, providing large
throughput and high SE is significantly important for future IoT networks.

As a key technology of 5G/B5G communications, NOMA is capable of achieving the massive connectivity in ultra-dense networks, and improving the SE and EE by allowing the user devices to share the same subchannel simultaneously \cite{ 7273963,7982794,7488207,8413117,7974749}.
\cite{7982794} showed that employing NOMA in massive MIMO networks can markedly
improve the SE, compared with those using OFDMA.
In \cite{7488207}, power allocation was optimized to maximize the EE of single-carrier multiuser NOMA systems.
By applying SIC and superposition coding at the WPT-enabled NOMA receiver \cite{8413117}, the secure rate and EE were enhanced, which can guarantee the information security and green communication.
By applying NOMA to the beamspace MIMO, \cite{7974749} enlarged the accessed device number on the same frequency band simultaneously, and attained a high SE.

Traditional design of massive MIMO-NOMA networks focused on improving the system throughput or SE \cite{8444982}. \cite{8637821} investigated a new multi-antenna non-orthogonal beamspace framework, which achieved the massive connectivity for cellular massive MIMO IoT with limited spectrum resources.
In \cite{6623072}, a joint optimization strategy of transmission power and transfer time, was proposed for massive antenna array systems to accomplish a long-distance and delay-guaranteed resource allocation.
Deploying multiple transmission antennas in massive MIMO-NOMA networks could potentially scale up the system SE, but this method resulted in expensive RF chain hardware cost, large energy consumption, and high signal processing complexity at transmitters \cite{8119827}.
The AS techniques provide a practical means to address such problem while guaranteeing the diversity and throughput gains of massive MIMO.
There are only a few studies focusing on using AS for NOMA systems \cite{7536954}. Regarding to the considered multicell WPT-powered distributed massive MIMO-NOMA networks with AS, no existing works have been presented to improve the system EE performance.

Most current works assumed that CSI is perfectly known at the transmitters when allocating the radio resources \cite{8350092,8629017,7273963,7982794,7488207}.
As channel estimation errors, inaccurate feedback and time-varying channel characteristic are inevitable, perfect CSI can hardly be obtained by the transmitters \cite{7009979}.
In \cite{7934461}, a joint optimization problem of user scheduling, power allocation and SIC decoding was addressed under the conditions of imperfect CSI and QoS requirements.
In \cite{8119791}, user scheduling and power allocation were alternately optimized to converge to the maximal EE under imperfect CSI, and it showed that the proposed algorithm can outperform OMA schemes.
Accurate CSI knowledge is essential in WPT enabled massive MIMO-NOMA networks, as it can contribute to both higher energy conversion efficiency and faster channel transmission rate.

To overcome the difficulties in massive connectivity and high EE in current resource-limited IoT environments, an energy-efficient resource allocation optimization problem is investigated in the distributed massive MIMO-NOMA networks deployed multiple rechargeable IoT devices.
Different from the works under perfect CSI assumption, we take the imperfect CSI into consideration.
The effect of AS on WPT-based massive MIMO-NOMA networks has also not yet been studied.
In light of these observations, the transmit power, time, AS and subcarrier allocation are jointly investigated for the considered network.
Different from conventional greedy and centralized algorithms using the Lagrange multiplier method \cite{8534441, 7062017, 8390925, 8680645}, we invoke the distributed ADMM approach \cite{8186925, add202101Wei2012} to iteratively solve the primal and dual subproblems of the convex optimization.
As a result, the maximal EE performance can be achieved in a distributed and computationally efficient manner.

\section{System model}

As shown in Fig. 1, we consider a distributed multicell massive MIMO-NOMA communication system with a large number of battery-powered IoT devices. There are $K$ cells, denoted by $\mathcal{K}=\{1,2, \ldots, K\}$.
Each cell has a BS at the cell center, and multiple rechargeable IoT devices randomly distributed within the cell.
The BS in the $k$-th cell (i.e., $BS_k$, $k\in \mathcal{K}$) has multiple antennas, and each IoT device has a single antenna.
The system has a bunch of available subcarrier set, denoted by $\mathcal{S}=\{1,2, \ldots, S\}$.
The total system bandwidth is $B$, and the subcarrier bandwidth is $B_s$ $(s \in \mathcal{S})$, where $B_s = B/S$.
In the $k$-th ($k\in \mathcal{K}$) cell, the $u$-th ($u \in \mathcal{U}=\{1,2, \ldots, U\}$) IoT device assigned to subcarrier $s$ is denoted as $IoT$-$D_{k, u, s}$.
Every BS contains a channel estimator (CE) to estimate the channel.
Let $\mathcal{M}=\{{M_1}, \ldots, {M_k}, \ldots, {M_K}\}$ be the candidate antenna number set and $\mathcal{N}=\{\mathbf{N}_1, \ldots, \mathbf{N}_k, \ldots, \mathbf{N}_K\}$ be the selected antenna vector set for all cells, where ${M_{k}}$ denotes the total maximum antenna number of $BS_k$ and $\mathbf{N}_{k,s}=\left[N_{k, 1, s}, N_{k, 2, s},\ldots,N_{k, u, s}\right]$ is the selected optimal antenna number vector in the $k$-th cell on subcarrier $s$. $N_{k, u, s}$ denotes the optimal antenna number for $IoT$-$D_{k, u, s}$.
Note that in the AS process, we only select the optimal antenna at each BS.
Thus, $\mathbf{N}_{k,n}$ can be simplified as $\mathbf{N}_{k}$.
In the considered new massive MIMO-NOMA network, the harvest-then-transmit protocol is employed, in which the BSs charge the IoT devices via WPT, and the IoT devices send messages to the BSs by using the energy harvested from WPT process.

In Fig. 2, a two-time-slot division protocol is developed for the WPT and WIT processes, which complies with the harvest-then-transmit protocol \cite{8350092}. 
Since there is no available power at IoT devices,
they replenish the power from the BSs to transmit their independent information.
A complete cycle of downlink and uplink transfer processes is treated as a time block, denoted by $T$.
Each block is partitioned into two time slots.
During the first slot $\tau_{k}$, $BS_k$ broadcasts wireless energy to $IoT$-$D_{k, u, s}$ via WPT technique,
and $IoT$-$D_{k, u, s}$ saves the harvesting energy in a chargeable battery. During the second slot $(T-\tau_{k})$, all IoT devices transmit information simultaneously by employing NOMA on subcarrier $s$.

\begin{figure}
\includegraphics[width=9cm]{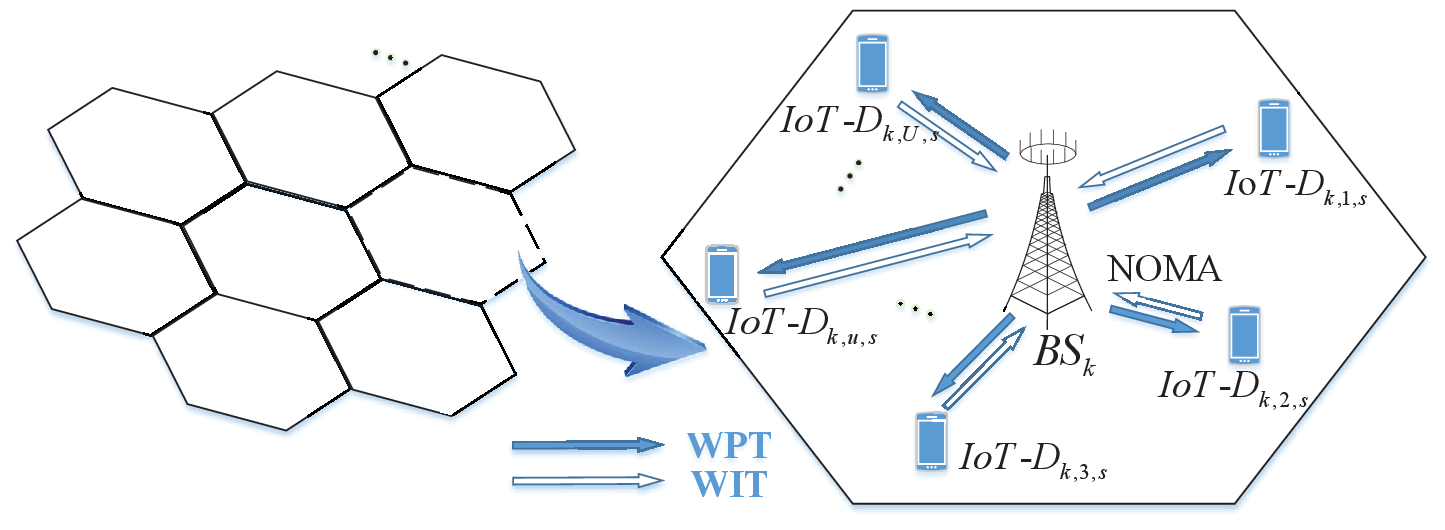}
\caption{A WPT enabled massive MIMO-NOMA network with massive IoT devices.}
\end{figure}
\begin{figure}
\includegraphics[width=6.5cm]{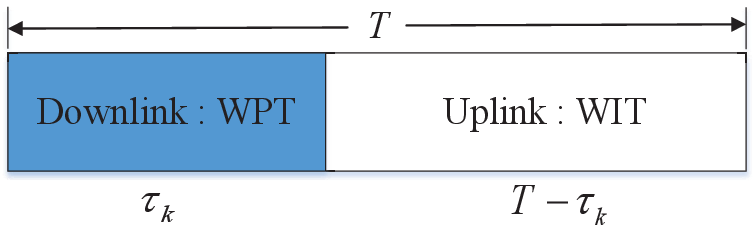}
\caption{Two-time-slot protocol for WPT and WIT.} \vspace{-1em}
\end{figure}

In our system, $BS_k$ is equipped with multiple antennas, which can reduce the EE.
To save the energy and enhance the EE of the system, we allocate the appropriate antenna number for $IoT$-$D_{k, u, s}$, denoted by $N_{k, u, s}$.
Meanwhile, the energy beamforming technique is employed for WPT to improve the power transfer efficiency.
In the rest of this section, we specify the EE for the cases of the perfect and imperfect CSI knowledge.
Notations of this paper are collated in Table I.

\begin{table}[h]
	\centering{}
    \label{table00}
	\textbf{Table \uppercase\expandafter{\romannumeral1}}~~SUMMARY OF THE KEY NOTATIONS. \vspace{2pt}\\
	\setlength{\tabcolsep}{0.5mm}{
		\begin{tabular}{lp{200pt}} \toprule
			Notations   & Meanings  \\
			\midrule
			$\mathcal{K}$     & Set of all the cells \\
            $K$     & Maximum number of cells  \\
            $k$     & Index to the cells  \\
            $\mathcal{U}$ & Set of IoT devices in the $k$-th cell \\
            $U$ & Number of IoT devices in the $k$-th cell\\
            $u$ & Index to the IoT devices in the $k$-th cell \\
            $\mathcal{S}$ & Set of subcarriers \\
            $s$ & Index to the subcarrier\\
            $S$ & Maximal number of the subcarrier\\
            $T$ & Time block of the downlink WPT and uplink WIT\\
            $\tau_{k}$ & Downlink WPT time slot in the $k$-th cell \\
            $T-\tau_{k}$ &Uplink WIT time slot in the $k$-th cell\\
            $BS_{k}$ & BS in the $k$-th cell\\
            $IoT$-$D_{k, u, s}$ &The $u$-th device in the $k$-th cell on subcarrier $s$ \\
            $\mathbf{b}_{k, u, s}$ & Energy beamforming vector \\
            $P_k$ & Transmit power of $BS_{k}$ \\
            $c_{k, u, s}$ & Subcarrier allocation indicator for $IoT$-$D_{k, u, s}$ \\
            $\mathbf{h}_{k, u, s}$ & Downlink channel between $BS_{k}$ and $IoT$-$D_{k, u, s}$\\
            $\hat{\mathbf{h}}_{k, u, s}$ & Estimated channel between $BS_{k}$ and $IoT$-$D_{k, u, s}$ \\
            $\Gamma_{{k, u, s}}$ & Co-channel interference under perfect CSI \\
            $B_{s}$ & Bandwidth of each subcarrier\\
            $\mathcal{M}$&Set of all maximal candidate antennas\\
            $M_k$ & Total antenna number of $BS_k$ in the $k$-th cell\\
            $\mathcal{N}$ & Set of the selected optimal antenna number vector \\
            $\mathbf{N}_k$& Selected antenna number vector of $BS_k$ \\
            $N_{k, u, s}$ & Optimal number of antennas for $IoT$-$D_{k, u, s}$\\
            $\alpha_{k,u}$ & Path loss between $BS_{k}$ and $IoT$-$D_{k, u, s}$\\
            $n_{k, u, s}^{'}$ & AWGN in the downlink\\
            $n_{k, u, s}^{''}$ & AWGN in the uplink\\
            $\eta$ &  Conversion efficiency of electrical energy\\
            ${\mathbf{e}}_{k, u, s}$ & Channel Estimation error between $BS_{k}$ and $IoT$-$D_{k, u, s}$\\
            $\Xi_{{k, u, s}}$ & Channel estimation and other IoT devices' interferences under imperfect CSI\\
            \toprule
	\end{tabular}}
\end{table} \vspace{-1.6em}

\subsection{Perfect Channel Condition}
With perfect CSI at the energy transmitter (i.e., the BS), the downlink received signal at $IoT$-$D_{k, u, s}$ on the subcarrier $s$ can be written as
\begin{equation}
\label{equa1}
y_{k, u, s}^{P}=\alpha_{k, u} \mathbf{b}_{k, u, s}^H \mathbf{h}_{k, u, s} z_{k}+n_{k, u, s}^{'}
\end{equation}
where $\mathbf{h}_{k, u, s}$ denotes the downlink channel gain vector between $BS_{k}$ and $IoT$-$D_{k, u, s}$, $\alpha_{k, u}$ is the coefficient of path loss, $\mathbf{b}_{k, u, s}$ denotes the energy beamforming vector, and $z_{k}$ is the energy signal transmitted by $BS_k$. $n_{k, u, s}^{'}$ denotes the additive white Gaussian noise (AWGN), which obeys the CSCG distribution, i.e., $n_{k, u, s}^{'} \sim \mathcal{CN} \left(0, \sigma^{2}\right)$.
The transmit power of $BS_k$ satisfies $E\left[\left|z_{k}\right|^{2}\right]=P_{k}$.
According to the law of energy conservation during the WPT process \cite{8474292}, we define the harvested energy of $IoT$-$D_{k, u, s}$ from $BS_k$ as
\begin{equation}
\label{equa2}
\begin{aligned}
E_{k, u, s}=\eta \tau_{k}\left(\alpha_{k, u}^{2}\left|\mathbf{b}_{{k, u, s}}^H \mathbf{h}_{{k, u, s}}\right|^{2} P_{k}\right)
\end{aligned}
\end{equation}
where $\eta \ (0\leq \eta \leq 1)$ denotes conversion efficiency from the RF power source to electrical energy.
With the optimal antenna number $N_{k, u, s}$, the dimension of the downlink channel gain vector $\mathbf{h}_{k, u, s}$ is $N_{k, u, s} \times 1$. The beamforming vector is $\mathbf{b}_{k, u, s}=\frac{\mathbf{h}_{k, u, s}}{ \| \mathbf{h}_{k, u, s} \|}$ and the energy beamforming direction can be appropriately adjusted according to the CSI knowledge at each time block.

In the WIT slot $(T-\tau_{k})$, $IoT$-$D_{k, u, s}$ uses the harvested energy to send information to $BS_{k}$. The uplink NOMA transmission protocol is employed to complete the message transmission. The multi-user detection and SIC are implemented in a fixed decoding order of the received signals at the $BS_{k}$ \cite{8632657}. 
In the $k$-th cell, all the IoT devices are sorted by their channel gains
in descending order, i.e., $\left|\mathbf{h}_{k, 1, s}\right| \geq\left|\mathbf{h}_{k, 2, s}\right| \geq \cdots \geq\left|\mathbf{h}_{{k, u, s}}\right| \geq \cdots \geq\left|\mathbf{h}_{k, U, s}\right|$.
$BS_{k}$ decodes the messages from the IoT devices in a descending order of channel gain.
Specifically, the signal with the best channel is decoded first while regarding the rest of IoT devices' signals as noise, and the signal with the worst channel is decoded finally with any interference from all other active IoT devices canceled.

In general, for the $u$-th IoT device, $BS_{k}$ deletes the received signals from the first IoT device to the $(u-1)$-th IoT device, and treats the signals from the $(u+1)$-th IoT device to the $U$-th IoT device as noise. The uplink received signal from $IoT$-$D_{k, u, s}$ at $BS_k$ is given by
\begin{equation}
\label{equa3}
y_{k, u, s}^{PU}=\sqrt{\frac{E_{{k, u, s}}}{T-\tau_{k}}} \alpha_{{k, u}} \mathbf{h}_{{k, u, s}}^{{H}} {x}_{{k, u, s}}+\Gamma_{{k, u, s}}+n_{k, u, s}^{''}
\end{equation}
where $n_{k, u, s}^{''}$ is the AWGN in the phase of uplink WIT, and $n_{k, u, s}^{''} \sim \mathcal{CN} \left(0, \sigma^{2}\right)$. ${x}_{{k, u, s}}$ is the signal from $IoT$-$D_{k, u, s}$, which is transmitted via power-domain NOMA with the power $\frac{E_{k, u, s}}{T-\tau_{k}}$.
$\Gamma_{{k, u, s}}$ denotes both the intra-cell and inter-cell co-channel interference signals, and it is written as
\begin{equation}
\label{equa4}
\begin{aligned}
& {{\Gamma }_{k, u, s}}=
\sum\limits_{{{u}_{1}}=u+1}^{U}{\sqrt{\frac{ {{E}_{k,u_1,s}} }{T-{{\tau }_{k}}}{{c}_{k,u_1,s}}}}{{\alpha }_{k,{{u}_{1}}}}\mathbf{h}_{{k,u_1,s}}^{{H}}{{{x}}_{k,u_1,s}} \\
 & +\sum\limits_{{{k}_{2}}=1,{{k}_{2}}\ne k}^{K}{\sum\limits_{u=1}^{U}{\sqrt{\frac{{{E}_{k_2,u,s}}}{T-{{\tau }_{{k}_{2} }}}c_{k_2,u,s}}}{{\alpha }_{{{k}_{2}},u}}\mathbf{h}_{{k_2,u,s}}^{{H}}{{{x}}_{k_2,u,s}}} \\
\end{aligned}
\end{equation}
where ${{c}_{k,u_1,s}}$ is the subcarrier allocation indicator for $IoT$-$D_{k,u_1,s}$. The SINR of $IoT$-$D_{k, u, s}$ on the subcarrier $s$ is denoted by \cite{8603758}
\begin{equation}
\label{equa5}
\begin{aligned}
{{\gamma }_{k, u, s}^P}=\frac{\frac{{{E}_{k, u, s}}}{T-\tau_{k}}\alpha _{k,u}^{2}{{\left| {{\mathbf{h}}_{k, u, s}} \right|}^{2}}}{ I_{k,u_1,s}+ I_{k_2,u,s} +{\sigma^2}}
\end{aligned}
\end{equation}
where $I_{k,u_1,s}$ and $I_{k_2,u,s}$ are the intra-cell and the inter-cell co-channel interference, respectively, as given by
\begin{align}
I_{k,u_1,s} &= \sum\limits_{{{u}_{1}}=u+1}^{U}{\frac{{{E}_{k,u_1,s}}}{T-{{\tau }_{k}}}c_{k,u_1,s}\alpha _{k,{{u}_{1}}}^{2}{{\left| {{\mathbf{h}}_{k,u_1,s}} \right|}^{2}}}, \\
I_{k_2,u,s} &= \sum\limits_{{{k}_{2}}=1,{{k}_{2}}\ne k}^{K}{\sum\limits_{u=1}^{U}{\frac{{{E}_{k_2,u,s}}}{T-{{\tau }_{{k}_{2} }}}c_{k_2,u,s}\alpha _{{{k}_{2}},u}^{2}{{\left| {{\mathbf{h}}_{k_2,u,s}} \right|}^{2}}}}.
\end{align}

\subsection{Imperfect Channel Condition}
In practical massive MIMO-NOMA communication networks, the errors in channel estimation, feedback and quantization result in imperfect CSI at the transmitters, degrading the system EE.
This paper investigates a quasi-static Rayleigh channel, where the channel between
$BS_{k}$ and $IoT$-$D_{k, u, s}$ remains unchanged for each time block $T$, and varies independently between different time blocks \cite{7070752}.
In each block, $BS_{k}$ uses the CE to estimate the channel by
minimum mean square error (MMSE) method. The estimated channel is denoted as
\begin{equation}
\label{equa9}
\begin{aligned}
\hat{\mathbf{h}}_{k, u, s}=\mathbf{h}_{k, u, s} - {\mathbf{e}}_{k, u, s}
\end{aligned}
\end{equation}
where ${\mathbf{e}}_{k, u, s}\sim\mathcal{CN}(0,\sigma_{{{e}}_{k, u, s}}^2 \mathbf{I}_{{\mathbf{e}}_{k, u, s}})$ denotes the channel estimation error, and $\mathbf{I}_{{\mathbf{e}}_{k, u, s}}$ stands for the identity matrix.
To efficiently harvest energy, the energy beamforming policy is designed as $\hat{\mathbf{b}}_{k, u, s}$, using the maximal ratio transmission (MRT) linear precoding technique.
Considering the beamforming strategy and practical (imperfect) CSI, the system adjusts the energy transfer direction, so that the IoT devices harvest as much energy as possible.
The received energy signal of $IoT$-$D_{k, u, s}$ in the WPT phase is written as
\begin{equation}
\label{equa10}
\begin{aligned}
y_{k, u, s}^{imP}=\alpha_{k, u} \hat{\mathbf{b}}_{k, u, s}^H {\mathbf{h}}_{k, u, s} z_{k}+n_{k, u, s}^{'}
\end{aligned}
\end{equation}
where $\hat{\mathbf{b}}_{k, u, s} = \frac{\hat{\mathbf{h}}_{k, u, s}}{ \| \hat{\mathbf{h}}_{k, u, s} \|}$.
The same as the perfect CSI case, the energy harvested by $BS_{k}$ under imperfect CSI knowledge is denoted by
\begin{equation}
\label{equa11}
\begin{aligned}
& E_{k, u, s}=\eta \tau_{k}\left(\alpha_{k, u}^{2}\left|\hat{\mathbf{b}}_{{k, u, s}}^H {\mathbf{h}}_{{k, u, s}}\right|^{2} P_{k}\right) \\
& \ \ \ \ \ \ \ \ = \eta \tau_{k} \alpha_{k, u}^{2} \varpi_{k, u, s} P_{k}
\end{aligned}
\end{equation}
where $\varpi_{k, u, s} $ is the degradation coefficient resulting from imperfect CSI, and $\varpi_{k, u, s}= \frac{\sigma^{2}_{e_{k, u, s}}}{1+\sigma^{2}_{e_{k, u, s}}}+\frac{\left\|\hat{\mathbf{h}}_{k, u, s}\right\|^{2}}
{\left({1+\sigma^{2}_{e_{k, u, s}}}\right)^{2}}$ \cite{8474292}.

In the phase of WPT, the received signal from $IoT$-$D_{k, u, s}$ to $BS_{k}$ can be expressed as
\begin{small}
\begin{equation}
\label{equa12}
\begin{aligned}
y_{k, u, s}^{imPU}\text{=}\sqrt{\frac{E_{{k, u, s}}}{T-\tau_{k}}} \alpha_{{k, u}} \hat{\mathbf{h}}_{{k, u, s}}^H {x}_{{k, u, s}}\text{+} \Xi_{{k, u, s}}\text{+}n_{k, u, s}^{''}
\end{aligned}
\end{equation}
\end{small}\\
where $ \Xi_{{k, u, s}}$ denotes the interferences caused by the channel estimation and other IoT devices, and can be given by
\begin{small}
\begin{equation}
\label{equa13}
\begin{aligned}
{{\Xi }_{k, u, s}}
&=
{\sqrt{\frac{{{E}_{k,u,s}}}{T-\tau_{k}}  }}{{\alpha }_{k,{{u}}}}\hat{\mathbf{e}}_{k,u,s}^H{{ {x}}_{k,u,s}}+\\
&\sum\limits_{{{u}_{1}}=u+1}^{U}{\sqrt{\frac{{{E}_{k,u_1,s}}}{T-\tau_{k}}{{c}_{k,u_1,s}}}}{{\alpha }_{k,{{u}_{1}}}}\hat{\mathbf{h}}_{k,u_1,s}^H {{ {x}}_{k,u_1,s}}+\\
 &\sum\limits_{{{k}_{2}}=1,{{k}_{2}}\ne k}^{K}{\sum\limits_{u=1}^{U}{\sqrt{\frac{{{E}_{k_2,u,s}}}{T-{{\tau }_{{{k}_{2}} }}}c_{k_2,u,s}}}{{\alpha }_{{{k}_{2}},u}}\hat{\mathbf{h}}_{k_2,u,s}^H {{ {x}}_{k_2,u,s}}}.
\end{aligned}
\end{equation}
\end{small}
\vspace{-1em}

To effectively complete accurate signal detection and decoding in the uplink, it is important to acquire CSI knowledge derived from the known training pilots which are transmitted in regular time intervals.
Once the BS receives the superimposed signals of the IoT devices,
several steps of signal processing will be employed to retrieve each IoT device's desired signal.
Generally, the message of IoT device with the highest channel gain is first decoded to guarantee that $BS_k$ detects this IoT device's message while treating other IoT devices' messages with low channel gain as noise.
For an ideal SIC, a replica of strong interference signal successfully detected by the BS is first generated before the BS subtracting it from the received signals, which are beneficial to cancel the effects from subsequent signal detection of IoT devices under weak channel conditions.
Nonetheless, affected by the channel noises and estimation errors, it is prone to generate the errors for the IoT devices' signal detection under strong channel conditions. Therefore, error propagation
while executing SIC is inevitable for the weak IoT devices. To alleviate the interference between IoT devices, the symbol level SIC and code word level SIC \cite{8771371} are jointly adopted to mitigate the error propagation under imperfect channel estimation.

For the imperfect case, we assume that $\left|\hat{\mathbf{h}}_{k, 1, s}\right| \geq \cdots \geq\left|\hat{\mathbf{h}}_{{k, u, s}}\right| \geq \cdots \geq\left|\hat{\mathbf{h}}_{k, U, s}\right|$.
As the signals of IoT devices with the strong channel gain are detected and decoded, the SIC can be effectively performed to cancel the inter-IoT-device interference.
The SINR of $IoT$-$D_{k, u, s}$ under imperfect CSI is written as \cite{8603758,8119791}
\begin{equation}
\label{equa14}
\begin{aligned}
{{\gamma }_{k, u, s}^{imP}}=\frac{\frac{{{E}_{k, u, s}}}{T-\tau_{k}}\alpha _{k,u}^{2}{{\left| {{\hat{\mathbf{h}}}_{k, u, s}} \right|}^{2}}}{ I_{k, u, s}^{e}+ I_{k,u_1,s}^{'}+ I_{k_2,u,s}^{'} +{\sigma^2}}
\end{aligned}
\end{equation}
where $I_{k, u, s}^{e}$ is the estimation noise since self-interference caused by the channel estimation error can be regarded as a noise. $I_{k,u_1,s}^{'}$ is intra-cell co-channel interference. $I_{k_2,u,s}^{'}$ is inter-cell co-channel interference. They can be expressed as
\begin{align}
 I_{k, u, s}^{e}
 &= {\frac{{{E}_{k,u,s}}}{T-\tau_{k}} \alpha _{k,{{u}}}^{2} \sigma^{2}_{e_{k, u, s}} },  \\
I_{k,u_1,s}^{'}
  &= \sum\limits_{{{u}_{1}}=u+1}^{U}{\frac{{{E}_{k,u_1,s}}}{T-{{\tau }_{k}}}c_{k,u_1,s}\alpha _{k,{{u}_{1}}}^{2}{{\left| {{\hat{\mathbf{h}}}_{k,u_1,s}} \right|}^{2}}},\\
  I_{k_2,u,s}^{'}
  &=\sum\limits_{{{k}_{2}}=1,{{k}_{2}}\ne k}^{K}{\sum\limits_{u=1}^{U}{\frac{{{E}_{k_2,u,s}}}{T-\tau_{{k}_{2}}}c_{k_2,u,s}\alpha _{{{k}_{2}},u}^{2}{{\left| {{\hat{\mathbf{h}}}_{k_2,u,s}} \right|}^{2}}}}.
\end{align} \vspace{-2em}

\subsection{Definition of the EE }
According to the channel hardening characteristic of the massive MIMO-NOMA networks and the Shannon capacity formula \cite{6725592}, the data rate from $IoT$-$D_{k, u, s}$ to $BS_{k}$ at the $s$-th subcarrier is given by
\begin{equation}
\label{equa8}
\begin{aligned}
R_{k, u, s}=B_s \log _{2}\left[1+\left(1+\ln \frac{M_{k}}{N_{k, u, s}} \right) \gamma_{k, u, s} {N_{k, u, s}} \right]
\end{aligned}
\end{equation}
where $\gamma_{k, u, s} = \gamma_{k, u, s}^P$ in the case of perfect CSI, and $\gamma_{k, u, s} = \gamma_{k, u, s}^{imP}$ in the case of imperfect CSI. The validity of (\ref{equa8}) under imperfect CSI is proved in Appendix A.

The total throughput of the considered system is denoted as
\begin{equation}
\label{equa18}
\begin{aligned}
R_\text{tot}\left( \boldsymbol{P}, \boldsymbol{\tau}, \boldsymbol{N}, \boldsymbol{C}\right)
\text{=}\sum_{k=1}^{K} \sum_{u=1}^{U} \sum_{s=1}^{S} c_{k, u, s} (T-\tau_{k}) R_{k, u, s}
\end{aligned}
\end{equation}
where $\boldsymbol{P}=[P_1,P_2,\ldots,P_K]$, $\boldsymbol{\tau}=[\tau_1,\tau_2,\ldots,\tau_K]$, $\boldsymbol{N}=[N_{k, u, s}]$ and $\boldsymbol{C}=[c_{k, u, s}] (\forall k\in\mathcal{K}, u\in\mathcal{U}, s\in\mathcal{S})$ are the vectors of transmit power, time, selected antenna and subcarrier allocation policies, respectively.
The total consumed energy is given by
\begin{small}
\begin{equation}
\label{equa19}
\begin{aligned}
&{E_{\text {tot}}\left( \boldsymbol{P}, \boldsymbol{\tau}, \boldsymbol{N}, \boldsymbol{C} \right)=}\\ &\sum_{k=1}^{K}\left( \left(P_\text{bs}  \max_u \{ N_{k, u, s} \}\text{+} U P_\text{user}\right) T\text{+}\sum_{s=1}^{S} P_{k} c_{k, u, s}   \tau_{k}  \right)
\end{aligned}
\end{equation}
\end{small}\\

\vspace{-2em}
\noindent where $P_\text{bs}$ and $P_\text{user}$ are the power consumptions of each antenna at $BS_{k}$ and each IoT device, respectively. $P_\text{bs}=P_\text{DAC}+P_\text{mix}+P_\text{filt}$  and $P_\text{user}=P_\text{syn}+P_\text{LNA}+P_\text{mix}+P_\text{IFA}+P_\text{filr}+P_\text{ADC}$ \cite{6725592,1321221}, where
$P_{\mathrm{DAC}}$ is the power consumption of the digital-to-analog converter;
$P_{\mathrm{mix}}$ is the power consumption of the mixer;
$P_{\mathrm{LNA}}$ is the power consumption of the low-noise amplifier;
$P_{\mathrm{IFA}}$ is the power consumption of the intermediate frequency amplifier;
$P_{\mathrm{filt}}$ is the power consumption of active filters at the transmitter;
$P_{\mathrm{filr}}$ is the power consumption of active filters at the receiver;
$P_{\mathrm{ADC}}$ is the power consumption of the analog-to-digital converter; and
$P_{\text{syn }}$ is the power consumption of the frequency synthesizer.

Generally, the system EE can be defined by information bit number reliably transmitted to all BSs per unit of the consumed energy. Thus, the EE of the considered massive MIMO-NOMA networks can be calculated as
\begin{equation}
\label{equa20}
\begin{aligned}
 {\eta_{\text{EE}}\left( \boldsymbol{P}, \boldsymbol{\tau}, \boldsymbol{N}, \boldsymbol{C} \right)}
=\frac{R_\text{tot}\left( \boldsymbol{P}, \boldsymbol{\tau}, \boldsymbol{N}, \boldsymbol{C} \right)}{E_\text{tot}\left(\boldsymbol{P}, \boldsymbol{\tau}, \boldsymbol{N}, \boldsymbol{C} \right)}.
\end{aligned}
\end{equation}
\section{Problem Formulation and Transformation}
\subsection{Problem Formulation}
In this section, an energy-efficient scheme of power, time, antenna and subcarrier allocation is investigated, which can achieve an optimal EE for the massive MIMO-NOMA networks. The EE maximization problem of the considered system is formulated as
\begin{subequations}
\label{equP1}
\begin{align}
& \ \ \ \ \mathbf{P 1}:\ \ \ \ \ \underset{\boldsymbol{P}, \boldsymbol{\tau}, \boldsymbol{N}, \boldsymbol{C} }{\mathop{\max }} \ \ \eta_\text{EE}\left( \boldsymbol{P}, \boldsymbol{\tau}, \boldsymbol{N}, \boldsymbol{C} \right) \\
& s.t.\mathbf{C1}:\ \ \ 0\le {{P}_{k}}\le {{P}_{\text{bs},\max }}, \ \  \forall k \in \mathcal{K} ,\\
&\ \ \ \ \mathbf{C2}:\ \ \ 0\le \tau_{k}\le T,  \ \ \forall k \in \mathcal{K} ,\\
&\ \ \ \ \mathbf{C3}:\ \ \ 0\le \frac{{{E}_{k, u, s}}}{T-\tau_{k}}\le {{P}_{\text{user},\max }}, \forall k \in \mathcal{K}, u \in \mathcal{U}, s \in \mathcal{S},\\
&\ \ \ \ \mathbf{C4}:\ \ \ {{{R}_{k, u, s}}} \ge {{R}_{\min }}, \forall k \in \mathcal{K}, u \in \mathcal{U}, s \in \mathcal{S},\\
&\ \ \ \ \mathbf{C5}:\ \ \ N_{k, u, s}\in  \{1, 2,...,{ M}_{k}\}, \forall k \in \mathcal{K}, u \in \mathcal{U}, s \in \mathcal{S},\\
&\ \ \ \ \mathbf{C6}:\ \ \ {c_{k, u, s}} \in \{0,1 \}, \sum_{u=1}^{U} c_{k, u, s} \leq U, \forall k \in \mathcal{K}, u \in \mathcal{U},  \nonumber \\
&\ \ \ \ \ \ \ \ \ \ \ \ \ \ \ \ \ \ \ \ \ \ \ \ \ \ \ \ \ \ \ \ \ \ \ \ \ \ \ \ \ \ \ \ \ \ \ \ \ \ s \in \mathcal{S},
\end{align}
\end{subequations}
where
$\mathbf{C1}$ specifies the maximal transmit power ${P}_{\text{bs},\max}$ of $BS_k$;
$\mathbf{C2}$ represents the range of the energy transfer time;
$\mathbf{C3}$ ensures that the transmit power of $IoT$-$D_{k, u, s}$ should be non-negative and no larger than ${{P}_{\text{user},\max }}$;
$\mathbf{C4}$ ensures that the channel rate of $IoT$-$D_{k, u, s}$ should be no lower than the minimal rate $R_{min}$;
$\mathbf{C5}$ controls the number of active antennas allocated for each IoT device to provide the fairness among IoT devices and save the RF cost; and
$\mathbf{C6}$ ensures subcarrier allocation, and one subcarrier can be multiplexed by at most $U$ IoT devices.

\subsection{Problem Transformation}
Problem $\mathbf{P1}$ is a fractional programming problem, which is non-convex.
To solve $\mathbf{P1}$, we first prove the concavity of $R_\text{tot}\left( \boldsymbol{P}, \boldsymbol{\tau}, \boldsymbol{N}, \boldsymbol{C} \right)$ (see Appendix B), and then use the Dinkelbach method \cite{WD} to convert it to a linear problem.
We provide the following theorem that presents the sufficient and necessary conditions for the optimal solution to $\mathbf{P1}$.

\newtheorem{theorem}{Theorem}
\begin{theorem}
The maximum EE can be achieved, if and only if
\begin{equation}
\begin{aligned}
&  \underset{\boldsymbol{P}, \boldsymbol{\tau}, \boldsymbol{N}, \boldsymbol{C}}{\mathop{\max }} \ \
\lbrace R_\text{tot}\left( \boldsymbol{P}, \boldsymbol{\tau}, \boldsymbol{N}, \boldsymbol{C} \right)-\eta_\text{EE}^*E_\text{tot}\left( \boldsymbol{P}, \boldsymbol{\tau}, \boldsymbol{N}, \boldsymbol{C} \right) \rbrace\\
& =R_\text{tot}\left( \boldsymbol{P}^*,\boldsymbol{\tau}^*,\boldsymbol{N}^*, \boldsymbol{C}^* \right)-\eta_\text{EE}^*E_\text{tot}\left( \boldsymbol{P}^*,\boldsymbol{\tau}^*,\boldsymbol{N}^*, \boldsymbol{C}^*  \right)=0,\\
\end{aligned}
\end{equation}
where $\eta_\text{EE}^*$ is the global optimal EE, as given by
\begin{equation}
\eta_\text{EE}^*=\frac{R_\text{tot}\left( \boldsymbol{P}^*,\boldsymbol{\tau}^*,\boldsymbol{N}^*, \boldsymbol{C}^*  \right)}{E_\text{tot}\left( \boldsymbol{P}^*,\boldsymbol{\tau}^*,\boldsymbol{N}^*, \boldsymbol{C}^*  \right)}.
\end{equation}
Herein, $\boldsymbol{P}^*$, $\boldsymbol{\tau}^*$, $\boldsymbol{N}^*$ and $\boldsymbol{C}^*$ denote the optimal power, WPT time, AS and subcarrier allocation policies, respectively.
\end{theorem}
\begin{proof}
See Appendix C.
\end{proof}

According to Theorem 1, $\mathbf{P1}$ can be converted to a subtractive linear form, which is given by
\begin{subequations}
\begin{align}
& \mathbf{P2}:\underset{\boldsymbol{P}, \boldsymbol{\tau}, \boldsymbol{N}, \boldsymbol{C}}{\mathop{\max }} \lbrace R_\text{tot}\left( \boldsymbol{P}, \boldsymbol{\tau}, \boldsymbol{N}, \boldsymbol{C} \right)-\eta_\text{EE}E_\text{tot}\left( \boldsymbol{P}, \boldsymbol{\tau}, \boldsymbol{N}, \boldsymbol{C} \right) \rbrace\\
& s.t.\ \ \ \ \mathbf{C1,C2,C3,C4,C5,C6}.
\end{align}
\end{subequations}

We further deal with the subtractive problem $\mathbf{P2}$, which, however, is a non-convex integer programming problem.
Solving this problem wound suffer a prohibitively high computational complexity.
To reduce the computation complexity, we solve it in the dual domain and obtain an equivalent solution (for e.g., when using the simplex method, if the original problem has many constraints and few variables, solving it in dual domain is feasible and efficient).
Although there is a non-zero duality gap between the primal problem and dual problem, \cite{1658226} has proved that it is negligible for such non-convex problem in the time-sharing multi-carrier system.
Since the number of subcarriers is large, e.g.,
from 32 to 256, the solution is asymptotically optimal.
According to \cite{1658226}, we relax $N_{k, u, s}$ and $c_{k, u, s}$ in $\mathbf{C5}$ and $\mathbf{C6}$ to be continuous values \cite{6364677,8474292}.
We define two new auxiliary continuous variables $N_{k, u, s}^{\dagger} \in \left[1, M_{k}\right]$ and $c_{k, u, s}^{{\dagger}} \in \left[0, 1\right]$ to replace the original variables $N_{k, u, s}$ and $c_{k, u, s}$.
$\mathbf{P2}$ can be transformed into $\mathbf{P3}$ as follows.
\begin{subequations}
\label{equp3}
\begin{align}
&\mathbf{P3}:\underset{\boldsymbol{P},\boldsymbol{\tau},\boldsymbol{N}^{\dagger},\boldsymbol{C}^{\dagger}}{\max} \lbrace R_\text{tot}(\boldsymbol{P},\boldsymbol{\tau},\boldsymbol{N}^{\dagger},\boldsymbol{C}^{\dagger} )-\eta_\text{EE}E_\text{tot}(\boldsymbol{P},\boldsymbol{\tau},\boldsymbol{N}^{\dagger},\boldsymbol{C}^{\dagger}   ) \rbrace\\
& \ \ s.t.\ \ \ \ \ \ \ \ \ \mathbf{C1}, \mathbf{C2}, \mathbf{C3}, \mathbf{C4}, \\
& \ \ \ \ \ \ \ \ \ \ \ \ \ \ \ \overset{\scriptscriptstyle\smile}{\mathbf{C5}}:N_{k, u, s}^{\dagger} \in \left[1, M_{k}\right], \\
& \ \ \ \ \ \ \ \ \ \ \ \ \ \ \ \overset{\scriptscriptstyle\smile}{\mathbf{C6}}:c_{k, u, s}^{{\dagger}} \in \left[0, 1\right], \sum_{u=1}^{U} c_{k, u, s}^{{\dagger}} \in \left[0, U\right].
\end{align}
\end{subequations}
By solving the optimal solution in the dual domain,
the asymptotically optimal solution for $\mathbf{P3}$ can be obtained according to the principle of strong duality theory \cite{add20201}.
The sub-optimality of AS can be given as $N_{k, u, s} = \left\lfloor N_{k, u, s}^{\dagger} \right\rfloor $.

As a result, we design a new iterative algorithm with convergence to obtain $\eta_\text{EE}^*$, which is given in Alg. 1.
The allocation strategy of power, time, AS and subcarrier allocation can be optimized by solving $\mathbf{P3}$.
As shown in Theorem 1, given a radio resource allocation strategy, e.g., $\{\boldsymbol{P}, \boldsymbol{\tau}, \boldsymbol{N}, \boldsymbol{C}\}$, we can find the corresponding solution for $\eta_\text{EE}$.
By running such iterative process, the optimal solution for resource allocation and $\eta_\text{EE}$ can be obtained.
We proceed to adopt ADMM method to decentralize and speed up the algorithm in next section.

\begin{algorithm}[h]
\caption{ An iterative algorithm for obtaining $\eta_{EE}^*$ by using Dinkelbach's method.}
\begin{algorithmic}[1]
\STATE \textbf{Initialization}: \\
  Initialize the maximum iteration index, $\Upsilon_{\max}$, the maximum tolerance $\varepsilon $, and the inner loop iteration index $\Upsilon = 1$.
\WHILE {(not converge) or $\Upsilon$ = $\Upsilon_{\max}$}
\STATE Solve the transformed $\mathbf{P3}$ in (\ref{equp3}) for given $\eta_{EE}$ and get the resource allocation policies $ \{ \boldsymbol{P}',\boldsymbol{\tau}',\boldsymbol{N}', \boldsymbol{C}' \}$. \
\IF  {$|{R_\text{tot}( \boldsymbol{P}',\boldsymbol{\tau}',\boldsymbol{N}', \boldsymbol{C}' )}-{\eta_{EE}}{E_\text{tot}( \boldsymbol{P}',\boldsymbol{\tau}',\boldsymbol{N}', \boldsymbol{C}' )}|\leq \varepsilon$}
\STATE Convergence = true;
\RETURN $\{ \boldsymbol{P}^*,\boldsymbol{\tau}^*,\boldsymbol{N}^*, \boldsymbol{C}^* \} = \{ \boldsymbol{P}',\boldsymbol{\tau}',\boldsymbol{N}', \boldsymbol{C}' \}$ and obtain the optimal ratio $\eta_{EE}^*$ by $\mathbf{Theorem 1}$; \
\ELSE
\STATE Convergence = false; \
\RETURN obtain a ratio of the total throughput to energy consumption $\eta_{EE} = \frac{R_\text{tot}( \boldsymbol{P}',\boldsymbol{\tau}',\boldsymbol{N}', \boldsymbol{C}' )}{E_\text{tot}( \boldsymbol{P}',\boldsymbol{\tau}',\boldsymbol{N}', \boldsymbol{C}' )}$ and $\Upsilon$ = $\Upsilon +1$;
\ENDIF
\ENDWHILE
\label{code:recentEnd}
\end{algorithmic}
\end{algorithm}
\section{Energy-Efficient Resource Allocation Via ADMM}

In this section, we propose a novel distributed ADMM-based resource allocation algorithm to divide the convex problem into several sequential optimization phases, which improves the efficiency of solving the problem and the quality of optimal solution.
To solve non-convex optimization problems, traditional methods first relax and convert it to a convex one, and then use the greedy and centralized algorithms based on Lagrange multiplier method \cite{8534441, 7062017, 8390925} to solve the converted problem.
These methods have relatively high iteration numbers due to the low efficiency of the algorithm execution, even only reaching a local optimum.

We first briefly review ADMM. Next, the applicability of distributed ADMM for this problem is illustrated. To efficiently solve $\mathbf{P3}$, we propose a distributed ADMM-based resource allocation algorithm for the considered network.

\subsection{Introduction to ADMM}
ADMM is an efficient algorithm, which can be applicable to distributed
convex optimization \cite{JEckstein, 6644242}. 
An important property of ADMM is that it can quickly converge
with a high accuracy of the optimal solution.
ADMM has been successfully applied in diverse optimization problems, e.g., statistical learning, multi-period portfolio optimization, network scheduling \cite{8186925}, \cite{JEckstein}.
The downside of the standard ADMM algorithm is that
the problem to be solved can only be partitioned into two subproblems. Standard ADMM cannot be implemented in a distributed manner for the massive MIMO-NOMA network.
For this reason, distributed ADMM is regarded as an effective approach, which can combine
the benefits of dual decomposition and augmented Lagrangian for multiple constraint optimization problems \cite{8186925,add202101Wei2012}.
By using the distributed ADMM, a global consensus problem is divided into some local subproblems. The solutions to local subproblems are further coordinated to obtain a consensus solution to the original problem.
\subsection{Transformation of $\boldsymbol{P3}$ }
We define several local variables, $\tilde\tau_{k}=\tau_{k}$, $\tilde{n}_{k, u, s}=N_{k, u, s}^{\dagger}$, and $\tilde{c}_{k, u, s}=c_{k, u, s}^{\dagger}$. By substituting the local variables into $\mathbf{P3}$, $\mathbf{P3}$ is equivalent to the following form
\begin{subequations}
\label{equp4}
\begin{align}
& \mathbf{P4}:\ \ \ \ \ \underset{\boldsymbol{P}, \tilde{\boldsymbol{\tau}},{\tilde{\boldsymbol{N}}},{\tilde{\boldsymbol{C}}}}{\max} \ \Lambda \left( \boldsymbol{P}, \tilde{\boldsymbol{\tau}},{\tilde{\boldsymbol{N}}},{\tilde{\boldsymbol{C}}}
 \right)
\\
& s.t.\ \ \ \ \ \ \ \tilde{\mathbf{C}}\mathbf{1} :\ \ \ 0 \leq P_{k} \leq P_{\text{bs}, \max },\\
& \ \ \ \ \ \ \ \ \ \ \ \tilde{\mathbf{C}}\mathbf{2} :\ \ \ 0 \leq \tilde\tau_{k} \leq T,\\
& \ \ \ \ \ \ \ \ \ \ \ \tilde{\mathbf{C}}\mathbf{3} :\ \ \ 0 \leq \frac{\tilde{E}_{k, u, s}}{T-\tilde\tau_{k}} \leq P_{\text{user}, \max },\\
& \ \ \ \ \ \ \ \ \ \ \ \tilde{\mathbf{C}}\mathbf{4} :\ \ \ \tilde{R}_{k, u, s} \geq R_{\min },\\
& \ \ \ \ \ \ \ \ \ \ \ \tilde{\mathbf{C}}\mathbf{5} : \quad  1 \leq \tilde{n}_{k, u, s} \leq M_{k},\\
& \ \ \ \ \ \ \ \ \ \ \ \tilde{\mathbf{C}}\mathbf{6} :\ \ \ 0 \leqslant \tilde{c}_{k, u, s} \leqslant 1, \sum\nolimits_{{u=1}}^{{U}}{{\tilde{c}_{k, u, s}}}  \leq U,
\end{align}
\end{subequations}

\noindent where $\Lambda\left( \boldsymbol{P}, \tilde{\boldsymbol{\tau}},{\tilde{\boldsymbol{N}}},{\tilde{\boldsymbol{C}}}\right)$ is the EE objective function after transformation, i.e., $ \Lambda\left( \boldsymbol{P}, \tilde{\boldsymbol{\tau}},{\tilde{\boldsymbol{N}}},{\tilde{\boldsymbol{C}}}\right) = R_\text{tot}\left( \boldsymbol{P}, \tilde{\boldsymbol{\tau}},{\tilde{\boldsymbol{N}}},{\tilde{\boldsymbol{C}}} \right)-\eta_\text{EE}E_\text{tot}\left( \boldsymbol{P}, \tilde{\boldsymbol{\tau}},{\tilde{\boldsymbol{N}}},{\tilde{\boldsymbol{C}}} \right)$.
With mathematical manipulations, the primal non-convex optimization problem $\mathbf{P1}$ can be converted to a convex problem $\mathbf{P4}$.
According to the property of the perspective function \cite{Boyd2004Convex}, $\mathbf{P4}$ is convex with respect to the variables $P_k$, $\tilde\tau_{k}$ and $\tilde{n}_{k, u, s}$.
The convexity of $\mathbf{P4}$ is proved in Appendix D.
As the transmit power $P_k$ concerning different $BS_{k}$ is correlated to some extent,
the coupling among all BSs urgently needs to be decoupled to solve $\mathbf{P4}$ by applying distributed ADMM.
We denote the local copy of $P_k$ as $\tilde{P}_{k}$, which is the perception of the global variables.
Through introducing the local variables $\tilde{P}_{k}$, $\tilde\tau_{k}$, $\tilde{n}_{k, u, s}$ and $\tilde{c}_{k, u, s}$, a feasible local variable set of $BS_{k}$ is given by
\begin{equation}
\mathbf{X_k}= \{\tilde{P}_{k}, \tilde\tau_{k}, \tilde{n}_{k, u, s}, \tilde{c}_{k, u, s} | \tilde{\mathbf{C2}}, \tilde{\mathbf{C3}}, \tilde{\mathbf{C4}},\tilde{\mathbf{C5}}, \tilde{\mathbf{C6}}\}.
\end{equation}

Furthermore, we define an associated local cost function with local variables $\tilde{P}_{k}$, $\tilde\tau_{k}$, $\tilde{n}_{k, u, s}$ and $\tilde{c}_{k, u, s}$, as given by
\begin{equation}
\begin{aligned}
  & \ \ \ {{g}_{k}}({{{\tilde{P}}}_{k}},{{{\tilde{\tau }}}_{k}}, \tilde{n}_{k, u, s}, \tilde{c}_{k, u, s}) =\\
   &\left\{
   \begin{aligned}
  & \sum\limits_{u=1}^{U}{\sum\limits_{s=1}^{S}{\tilde{c}_{k, u, s}{{{\tilde{R}}}_{k, u, s}}}}-\eta _{k}^{*}\left( ({{P}_{\text{bs}}} \underset{u}{\mathop{\max}}\,\{ \tilde{n}_{k, u, s}\} \text{+}U{{P}_{\text{user}}}) T \right. \\
  & +{{{\tilde{P}}}_{k}}\sum\limits_{s=1}^{S} \left. {\tilde{c}_{k, u, s}\,{{{\tilde{\tau }}}_{k}}}\right), \text{if}\ \tilde{{P}_{k}},{{{\tilde{\tau }}}_{k}},\tilde{n}_{k, u, s},\tilde{c}_{k, u, s}\in {{\mathbf{X}}_{\mathbf{k}}}, \\
 & \infty ,\ \ \ \text{otherwise.} \\
\end{aligned} \right. \\
\end{aligned}
\end{equation}

As a consequence, the global consensus problem $\mathbf{P4}$ can be transformed into $\mathbf{P4}$, as shown\\
\vspace{-2em}
\begin{subequations}
\label{equp5}
\begin{align}
& \mathbf{P5}:\ \ \ \ \ \ \ \underset{\tilde{\boldsymbol{P}}, \tilde{\boldsymbol{\tau}}, \tilde{\boldsymbol{N}}, {\tilde{\boldsymbol{C}}} }{\min} \ \sum_{k=1}^{K} g_{k}\left(\tilde{P}_{k}, \tilde\tau_{k}, \tilde{n}_{k, u, s}, \tilde{c}_{k, u, s}\right)
\\
& s.t.\ \ \ \ \ \ \ \ \ \widehat{\mathbf{C}}\mathbf{1} : \tilde{P}_{k}=P_{k},\ \ 1 \leq k \leq K.
\end{align}
\end{subequations}

\subsection{Energy-efficient Solution via ADMM}
The proposed algorithm, Alg. 1, is extended to achieve the efficient resource allocation of multicell massive MIMO-NOMA networks in a fully distributed fashion.
Inspired by \cite{8186925}, $\mathbf{P5}$ is a global consensus problem, whose solution can be obtained by using ADMM approaches in a distributed manner.
The mechanism of ADMM starts with constructing an augmented Lagrangian function concerning the consensus constraints \cite{8186925}. 
Not only does the augmented Lagrangian include a set of consensus
constraints weighted by the Lagrangian multipliers (i.e., the traditional
Lagrangian), but also adds a regularized
quadratic term: a squared $\text{L}_2$ norm of the consensus constraints.
We use $\lambda_{k}$ to denote the Lagrangian multiplier corresponding to the $k$-th consensus constraint of (\ref{equp5}), the augmented Lagrangian of (\ref{equp5}) is given by
\begin{equation}
\label{augme}
\begin{aligned}
  & \ \ \ \ {{\mathcal{L}}_{\rho }}\left( \{{{{\tilde{P}}}_{k}},{{{\tilde{\tau }}}_{k}},\tilde{n}_{k, u, s},\tilde{c}_{k, u, s}\},\{{{P}_{k}}\},\text{ }\!\!\{\!\!\text{ }{{\lambda }_{k}}\text{ }\!\!\}\!\!\text{ } \right) \\
 & =\sum\limits_{k=1}^{K}{{{g}_{k}}\left( {{{\tilde{P}}}_{k}},{{{\tilde{\tau }}}_{k}},\tilde{n}_{k, u, s},\tilde{c}_{k, u, s} \right)}+\sum\limits_{k=1}^{K}{{{\lambda }_{k}}\left( {{{\tilde{P}}}_{k}}-{{P}_{k}} \right)} \\
 &\ \ +\frac{\rho }{2}\sum\limits_{k=1}^{K}{{{\left( {{{\tilde{P}}}_{k}}-{{P}_{k}} \right)}^{2}}}
\end{aligned}
\end{equation}
where $\rho \in \mathbb{R}^{+}$ denotes a positive penalty parameter which has the function of adjusting the convergence rate of ADMM approach. $\rho$ also provides strict convexity with respect to global and local variables in (\ref{augme}), and thus problem (\ref{equp5}) is solvable \cite{7968315}.
As the structure of multicell massive MIMO-NOMA is interconnected,
the augmented Lagrangian is separable among the BSs.
Then, we execute the central parts of ADMM, which are used to update the global variables (i.e., $\boldsymbol{P}$), the local variables (i.e.,$\{\tilde{\boldsymbol{P}}, \tilde{\boldsymbol{\tau}},{\tilde{\boldsymbol{N}}},{\tilde{\boldsymbol{C}}}\}$), and the Lagrange multipliers (i.e., $\boldsymbol{\lambda}$) by using the Gauss-Seidel method.
In each iteration of executing ADMM, all BSs share the knowledge of interrelated
inter-cell interference temperatures and consensus EE through updating
the global variables. Then, each BS independently solves its
own subproblem and updates Lagrange multipliers to synchronize involved variables.
The sequential optimization steps of the updated procedures are given by
\begin{align}
& \{P_k\}^{[t+1]} := \nonumber\\
  & \arg\underset{\boldsymbol{P}}\min \left\{\sum_{k=1}^{K} \lambda_{k}^{[t]}\left(\tilde{P}_{k}^{[t]}-P_{k}\right)+\frac{\rho}{2} \sum_{k=1}^{K}\left(\tilde{P}_{k}^{[t]}-P_{k}\right)^{2}\right\}, \label{equpstep2}\\
&
 \left\{\tilde{P}_{k}, \tilde{\tau}_{k}, \tilde{n}_{k, u, s}, \tilde{c}_{k, u, s}\right\}^{[t+1]} := \nonumber \\
 & \arg\underset{
 \tilde{\boldsymbol{P}}, \tilde{\boldsymbol{\tau}},{\tilde{\boldsymbol{N}}},{\tilde{\boldsymbol{C}}}}\min
 \left\{g_{k}\left(\tilde{P}_{k}, \tilde\tau_{k}, \tilde{n}_{k, u, s}, \tilde{c}_{k, u, s}\right)+\lambda_{k}\left(\tilde{P}_{k}-P_{k}^{[t+1]}\right)\right. \nonumber\\
 &\left. +\frac{\rho}{2}\left(\tilde{P}_{k}-P_{k}^{[t+1]}\right)^{2}\right\}, \label{equpstep1}\\
 & \left\{\lambda_{k}\right\}^{[t+1]}:= \left\{\lambda_{k}\right\}^{[t]}+\rho\left(\tilde{P}_{k}^{[t+1]}-P_{k}^{[t+1]}\right) \label{equpstep3}
\end{align}
where the superscript $t$ indicates the $t$-th iteration.

\subsection{Distributed ADMM-based Resource Allocation Algorithm}

As illustrated in Alg. 2, a distributed ADMM-based resource allocation algorithm is summarized.
The algorithm consists of three steps: (a) update of global variables, (b) update of local variables, and (c) update of Lagrange multipliers.
Generally, in a relatively long period (or one iteration instead),
all the BSs first update the global variables to keep a consistent interference level and the consistent EE.
After that, each BS independently addresses the self subproblems, such as AS and subcarrier allocation.
Then, by collecting and using the local CSI of IoT devices, Lagrange multipliers are updated at the BSs.
Once a local convergence is reached, the BSs exchange the values of power with each other to attain the optimal solution and make proper adjustments for local resource allocation in the next round.
Through such an interactive and iterative process, the signaling overhead can decrease. No CSI needs to be exchanged between the BSs and IoT devices. Moreover, a large problem is decoupled into several distributed small problems to be simultaneously solved, hence remarkably improving the efficiency of solving the optimization problem.
Only the EE of each BS needs to be exchanged to run Alg. 1. For Alg. 2, we have the following key steps to elaborate.

\begin{algorithm}[h]
\caption{A distributed ADMM-based resource allocation algorithm.}
\begin{algorithmic}[1]
\STATE \textbf{Initialization}: \\
  a) Initialize the parameters: $K$, $U$, $M_{k }$, $\alpha_{k,u}$, $P_{\text{bs,max}}$, $P_{\text{user,max}}$, $R_\text{min}$, $P_\text{bs}$, $P_\text{user}$, $\eta_\text{EE}$; \\
  b) Initialize the transmit power $P_k$, the Langrange multiplier $\lambda_k$ and a small enough stop criterion threshold $\epsilon$; \\
  c) Collect the CSI of IoT devices in the coverage range of $BS_k$ by MMSE method to estimate their channels.\\
\FOR{t=1,2,...,G}
\STATE Initialize $P_{k}^{[\text{t}]}$ and $\lambda_{k}^{\text{[t]}}$ of $BS_{k}$. \
\STATE $BS_{k}$ updates $P_{k}^\text{[t+1]}$  according to (\ref{equpstep2}) through an average consensus algorithm \cite{add2020410}. \
\STATE $BS_{k}$ receives the results of $P_{k}^\text{[t+1]}$ from each $BS_{k}$ and updates $\left\{\tilde{P}_{k}, \tilde{\tau}_{k}, \tilde{n}_{k, u, s}, \tilde{c}_{k, u, s}\right\}^{[t+1]}$ by taking (\ref{equpstep1}). \
\STATE Update $\lambda_{k}^\text{[t+1]}$ via taking the iterative projection method in (\ref{equpstep3}). \
\IF { $\epsilon $ in (\ref{equa34}) is sufficiently small and satisfied}
\STATE {Go to Step 11};
\ENDIF
\ENDFOR
\STATE \textbf{Output}: The optimal resource allocation policy $\left\{\tilde{P}_{k}, \tilde{\tau}_{k}, \tilde{n}_{k, u, s}, \tilde{c}_{k, u, s}\right\}^{[t+1]}$.
\label{code:recentEnd}
\end{algorithmic}
\end{algorithm}

1)\ \  Update of global variables $\{P_k\}_{k \in \mathcal{K}}$

The iteration process of distributed ADMM-based algorithm begins with
the update of the global variables, which can be handled by fixing
the local variables when solving the problem
$\underset{\boldsymbol{P}} \min \ {{\mathcal{L}}_{\rho }}\left( \{{{{\tilde{P}}}_{k}},{{{\tilde{\tau }}}_{k}},\tilde{n}_{k, u, s},\tilde{c}_{k, u, s}\},\{{{P}_{k}}\},\text{ }\!\!\{\!\!\text{ }{{\lambda }_{k}}\text{ }\!\!\}\!\!\text{ } \right)$. Note that in (\ref{equpstep2}), $\tilde{P}_{k}^{[t]}$ is the value obtained
after the $t$-th iteration and is exchanged by all BSs.

2)\ \  Update of local variables $\{\tilde{P}_{k}, \tilde\tau_{k}, \tilde{n}_{k, u, s}, \tilde{c}_{k, u, s}\}_{k \in \mathcal{K}}$

For Step 5, the problems of finding the local selected antenna number and the resource allocation
strategies are separable among different BSs. Thus, the update of $\{\tilde{P}_{k}, \tilde\tau_{k}, \tilde{n}_{k, u, s}, \tilde{c}_{k, u, s}\}_{k \in \mathcal{K}}$ is decomposed into $K$ subproblems.
Each subproblem can be handled locally by each BS. The optimization subproblem at the $(t+1)$-th iteration in (\ref{equpstep1}) can be tackled in parallel by the BS. Specifically, $BS_k$ is able to locally solve the convex subproblem of (\ref{equpstep1}) only using its local CSI ${\mathbf{h}}_{{k, u, s}}$ and $\{P_k\}^{[t+1]}$ from the update of global variables.

3)  Update of Lagrange Multiplier $\{\lambda_k\}_{k \in \mathcal{K}}$

The last step of ADMM is to update the Lagrange
multiplier for Step 6 by using (\ref{equpstep3}). As the current iterations of the local variables and global variables are available to all BSs, updating the Lagrange multiplier in this step does not require additional information exchange, and therefore will not produce additional signaling overhead.

updating
the Lagrange multiplier does not require additional information exchange, and thus incurs no signaling overhead.

4)  Stop criteria and convergence

The iteration stops if
\begin{equation}
\label{equa34}
\begin{aligned}
\left\|\mathbf{r}_{k}^{[t+1]}\right\|_{2}^{2}=\left\|\tilde{\mathbf{P}}_{k}^{[t+1]}-\mathbf{P}_{k}^{[t+1]}\right\|_{2}^{2} \leq \epsilon, \forall {k} \in \mathcal{K}
\end{aligned}
\end{equation}
where $\mathbf{r}_{k}$ denotes the residual between the local and global variables.
For Step 9, the proposed distributed ADMM-based energy-efficient algorithm guarantees the
residual convergence, objective convergence and dual
variable convergence if $t \rightarrow \infty$. As the objective function of $\mathbf{P5}$ is a closed convex form, the Lagrangian ${{\mathcal{L}}_{\rho }}$ has a saddle point.

After obtaining the optimal solution, the antenna number is approximated as $N_{k, u, s} = \left\lfloor N_{k, u, s}^{\dagger} \right\rfloor $ by considering that the number of selected antennas is an integer.
The complexity of Alg. 2 during each iteration loop is influenced and charged by the complexity of solving subproblem (\ref{equpstep1}) of $\mathbf{P5}$ at each BS.
The polynomial time-complexity of addressing this problem is related to the number of IoT devices, i.e., $ \mathcal{O}(U)$.
The non-ergodic convergence rate of Alg. 2 is $\mathcal{O}(1/t)$, where $t$ is the index of iterations. In $t$-th iteration of the loop, (\ref{equpstep2}), (\ref{equpstep1}) and (\ref{equpstep3}) are updated sequentially with the computational complexity of $\mathcal{O}(KUS \text{log}_2({U^2/2+2KU}))$. Therefore, the whole computational complexity of Alg. 2 is $\mathcal{O}(KUS \text{log}_2({U^2/2+2KU})t)$.
\section{Simulation Results}

In this section, the EE performance of the proposed distributed ADMM-based energy-efficient resource allocation algorithm (denoted as ``PA'') is presented. We consider multiple cells, and the radius of each cell is set as 500 m. The system bandwidth is normalized with $B = 1$ Hz and the bandwidth of each subcarrier is $B_s = B/S$ Hz.
In the massive MIMO-NOMA networks, the IoT devices are randomly and uniformly distributed within each cell. 
Both path loss and Rayleigh fading are considered. 
The simulation results are based on averaging 10,000 iterations.
According to \cite{8390925} and \cite{1321221}, some simulation parameters are listed in Table \uppercase\expandafter{\romannumeral2}.

\begin{table}[h]
	\centering{}
    \label{table1}
	\textbf{Table \uppercase\expandafter{\romannumeral2}}~~SIMULATION PARAMETERS.\\ \vspace{0.5em}
	\setlength{\tabcolsep}{2mm}{
		\begin{tabular}{ll|ll} \toprule
			Parameter   & Value  & Parameter   & Value\\
			\midrule
			$K$     & 6     & $ \eta $ &  0.8          \\
			$U$     & 15    & $ \epsilon $ & $10^{-7}$         \\
            $S$    &  20 $\sim$ 40    &  $ B $ &  $1$ Hz        \\
            $P_{\text{bs,max}}$   & 46 dBm  & $ P_\text{DAC} $ & 10 mW \\
            $P_{\text{user,max}}$   & 23 dBm & $ P_\text{ADC} $ & 10 mW \\
            $R_\text{min}$   & 0.1 bit/s/Hz & $P_\text{filr}$ & 2.5 mW \\
            ${\alpha_{k,u}}$   & $1/75$ & $P_\text{filt}$ & 2.5 mW\\
            $P_\text{mix}$   & 30.3 mW& $P_\text{syn}$ & 50 mW\\
            $P_\text{LNA}$   & 20 mW& $P_\text{IFA}$ & 3 mW\\
            \toprule
	\end{tabular}} \vspace{-1em}
\end{table}

To show the effects of AS, channel information and transmission conditions, the performance of AS under imperfect CSI through NOMA mode is presented in Fig. 3 by plotting the EE versus the communication distance. We use ``OTS1'' and ``OTS2'' to denote the OMA (i.e., OFDMA) transmission scheme with AS and perfect CSI \cite{6364677}, and the OMA transmission scheme without AS \cite{6623072}, respectively.
In Fig. 3, we see that with the increase of the distance between $BS_{k}$ and $IoT$-$D_{k, u, s}$, the system EE decreases.
The reason is that as the distance increases,
$BS_k$ needs to transmit more power to provide the IoT devices to guarantee QoS.
Fig. 3 also shows that the EE under perfect CSI is higher than that under imperfect CSI, due to the impact of CSI and feedback errors.
The EE performance with AS is superior to the case with no AS.
We also compare the EE performance between NOMA and OMA transmissions under the consideration of AS and CSI.
We can see that the proposed distributed ADMM algorithm via NOMA outperforms ``OTS1'' and ``OTS2'' in terms of EE.

In Fig. 4, we compare the EE of ``PA'' with that of ``OTS1'' and ``OTS2'' under perfect and imperfect CSI.
We can observe that the EE of ``PA'' under perfect CSI is superior to all other cases thanks to its high-efficient AS scheme and perfect channel conditions.
Considering the estimation error in the practical channel environment, the EE of ``PA'' under imperfect CSI is lower than that under perfect CSI.
The EE of OMA-based ``OTS2'' with no AS and imperfect CSI has the worst performance.
This is because it does not adopt the AS technique or NOMA transmission.
In terms of transmission mode, the EE of ``PA'' based on NOMA outperforms the OMA-based ``OTS1'' and ``OTS2'' whether the CSI is perfect or not.
This is because power domain NOMA technique can improve the accessing number of IoT devices and increase the SE and EE by employing the multiple access on the same time-frequency resource block.

\begin{figure*}
\begin{minipage}[t]{0.5\linewidth}
\includegraphics[width=8.4cm]{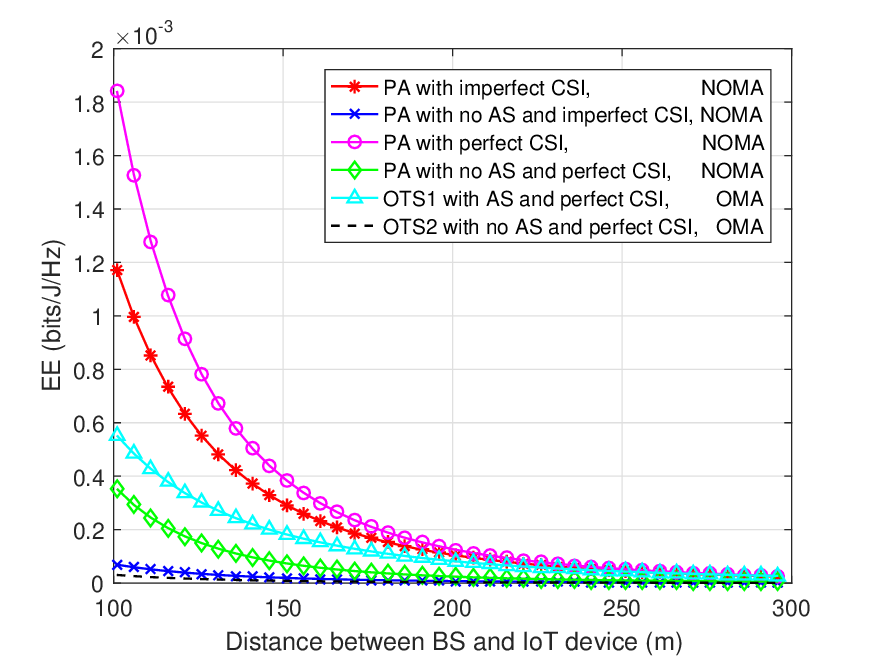}\vspace{-0.5em}
\caption{EE comparison between PA and reference schemes with different AS, CSI and transmission conditions.}
\end{minipage}  \vspace{-0.5em}
\begin{minipage}[t]{0.5\linewidth}
\includegraphics[width=8.4cm]{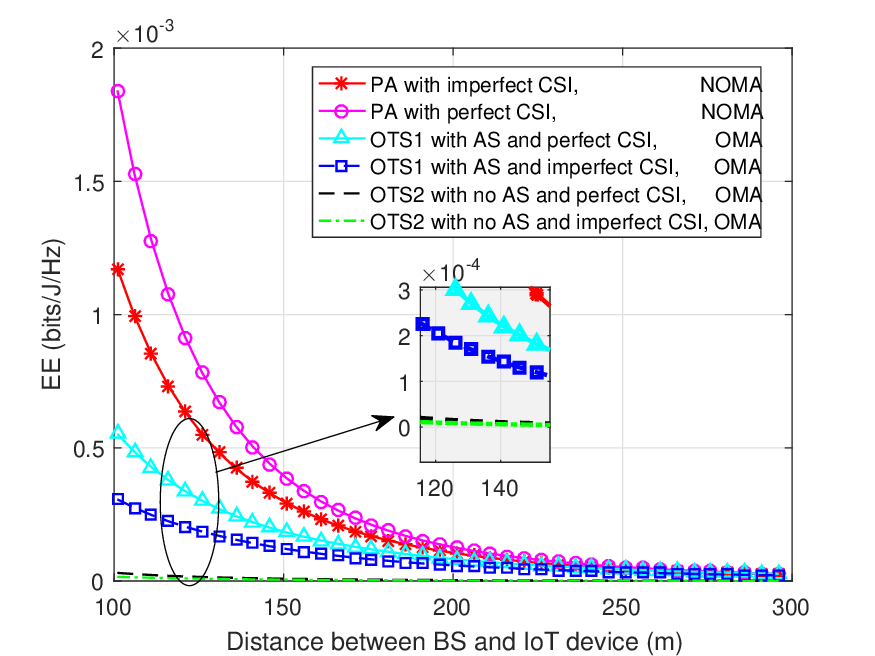}\vspace{-0.5em}
\caption{Comparison between NOMA-based PA and OMA-based benchmark schemes under perfect and imperfect CSI.} \vspace{-1em}
\end{minipage} \vspace{-0.5em}
\end{figure*}

Fig. 5 illustrates the influence of imperfect CSI on the system EE under different channel estimation parameters.
We set three values for the variance of channel estimation error, i.e., $\sigma_{e_{k, u, s}}^2 = 0.7$, $\sigma_{e_{k, u, s}}^2 = 0.5$, and $\sigma_{e_{k, u, s}}^2 = 0.3$.
The results of the proposed distributed resource allocation algorithm under perfect CSI are presented for comparison.
We can observe in Fig. 5 that, as the distance between $BS_k$ and $IoT$-$D_{k, u, s}$ increases, the system EE of the proposed distributed algorithm declines moderately for all the cases.
We can also see that Alg. 2 can obtain the optimal EE when the channel is perfect.
Fig. 5 shows the performance difference between the perfect and
imperfect CSI cases.
When the transmission distance is 100 m, the system EE under perfect CSI is approximately 3.5 times higher than that under imperfect CSI with $\sigma_{e_{k, u, s}}^2 = 0.7$.
When the transmission distance increases to 200 m, the EE under perfect CSI is nearly 1.5 times larger than that with $\sigma_{e_{k, u, s}}^2 = 0.7$.
As the channel estimation variance enlarges, the gap between different variances becomes large.
For instance, the EE with $\sigma_{e_{k, u, s}}^2 = 0.3$ is higher than the other two cases, i.e., the EE with $\sigma_{e_{k, u, s}}^2 = 0.5$ and $\sigma_{e_{k, u, s}}^2 = 0.7$. All the observations in Fig. 5 show that
the system EE depends on the accuracy of the obtained CSI.

Fig. 6 shows the EE versus the number of antenna under perfect and imperfect CSI.
We can see that, the EE first increases and then decreases with the increasing antenna number.
For example, the EE curve under imperfect CSI ($\sigma_{e_{k, u, s}}^2 = 0.6$) grows approximately from $0.6 \times 10^{-3} $ bits/J/Hz with $M_{k} = 10$, to $2.2 \times 10^{-3} $ bits/J/Hz with $M_{k} = 30$.
After reaching the peak at $2.2 \times 10^{-3}$ bits/J/Hz, the system EE of ``PA'' shows a downward trend.
The reason is that an appropriate antenna number can effectively improve the system EE, but excessive antennas would result in a rapidly increasing power consumption and reduce the EE.
Moreover, AS can alleviate hardware complexity of multi-antenna systems.
Therefore, the optimal antenna number
needs to be reasonably determined.


Fig. 7 shows the EE with different antenna numbers and CSI estimation errors.
As shown in Fig. 7, there is an optimal transmit power for a given antenna number and CSI to achieve the maximum EE.
As the transmit power increases, the system EE ascends slightly, then arrives at its maximum and lastly descends moderately under either perfect or imperfect CSI.
This shows that proper allocation of transmit powers can improve the system EE.
Fig. 7 also shows that, under different antenna numbers, the optimal EE is
different.
When the antenna number $N_{k, u, s}$ = 30, the EE optimum is $3.2 \times 10^{-3}$ bits/J/Hz, which is higher than other three cases, i.e., $N_{k, u, s}$ = 20, $N_{k, u, s}$ = 25, and $N_{k, u, s}$ = 40.
We also compare the system EE under perfect and imperfect CSI.
The results demonstrate that the EE under perfect CSI is superior to that under the imperfect CSI, which is consistent with Figs. 3--6.

\begin{figure*}
\begin{minipage}[t]{0.5\linewidth}
\centering\includegraphics[width=8.5cm]{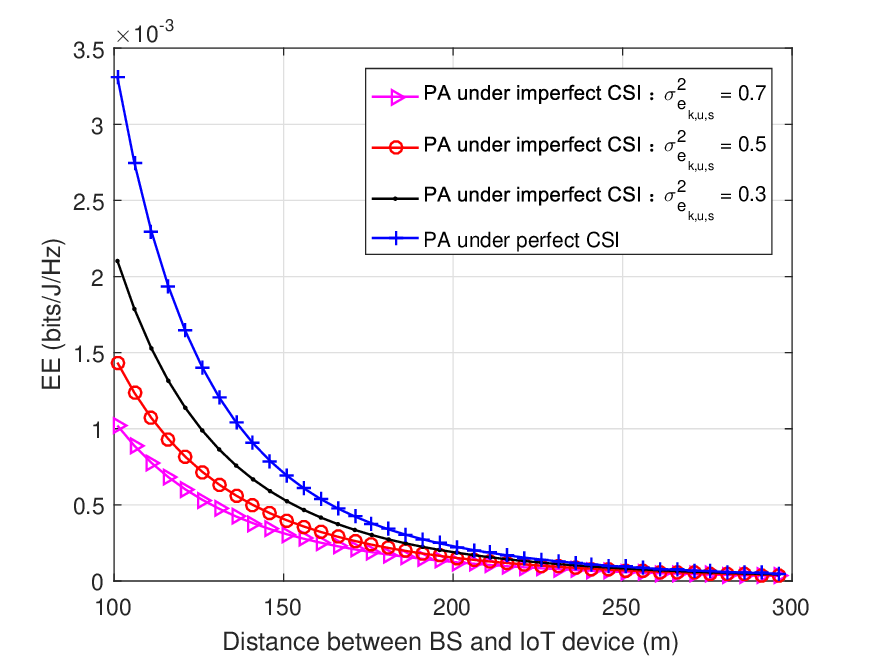} \vspace{-0.5em}
\caption{EE comparison of PA under perfect and imperfect CSI.}
\end{minipage} \vspace{-0.5em}
\begin{minipage}[t]{0.5\linewidth}
\centering\includegraphics[width=8.5cm ]{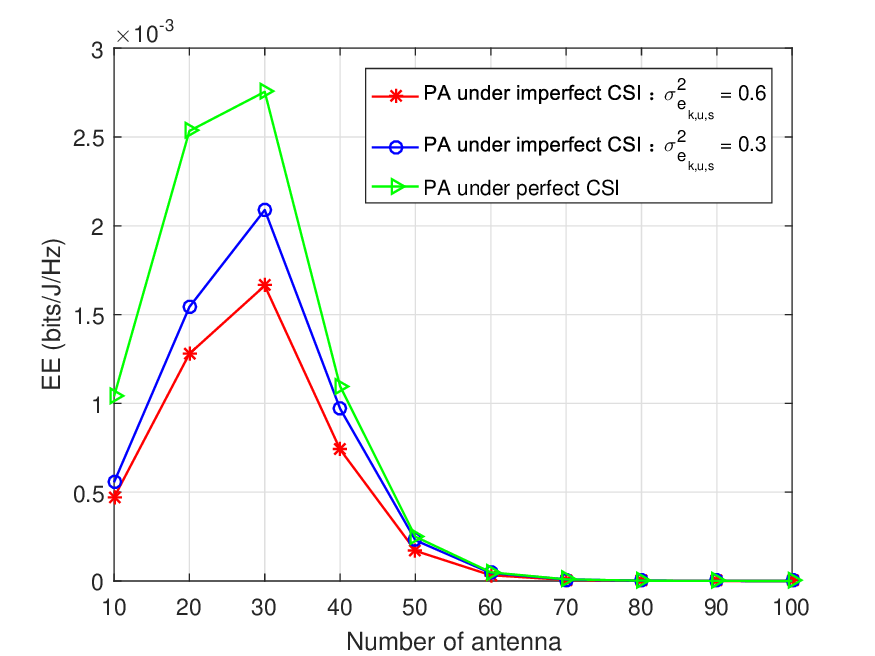}\vspace{-0.5em}
\caption{EE versus number of antenna under perfect and imperfect CSI.}
\end{minipage} \vspace{-0.5em}
\end{figure*}

\begin{figure*}
\begin{minipage}[t]{0.5\linewidth}
\centering\includegraphics[width=8.5cm]{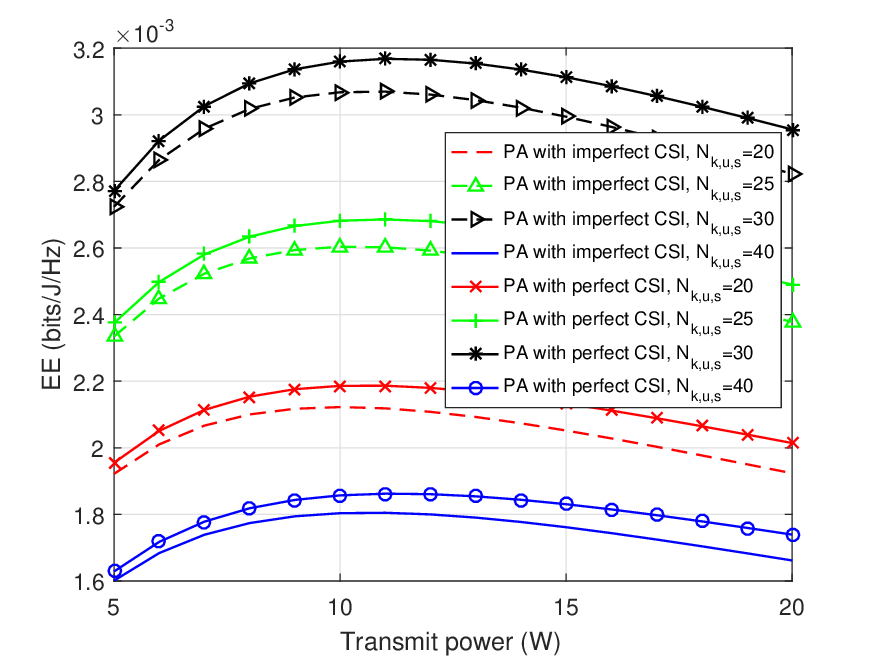} \vspace{-0.5em}
\caption{Impact of transmit power on EE with different number of antenna and CSI estimation errors.}
\end{minipage}  \vspace{-0.5em}
\begin{minipage}[t]{0.5\linewidth}
\centering\includegraphics[width=8.5cm ]{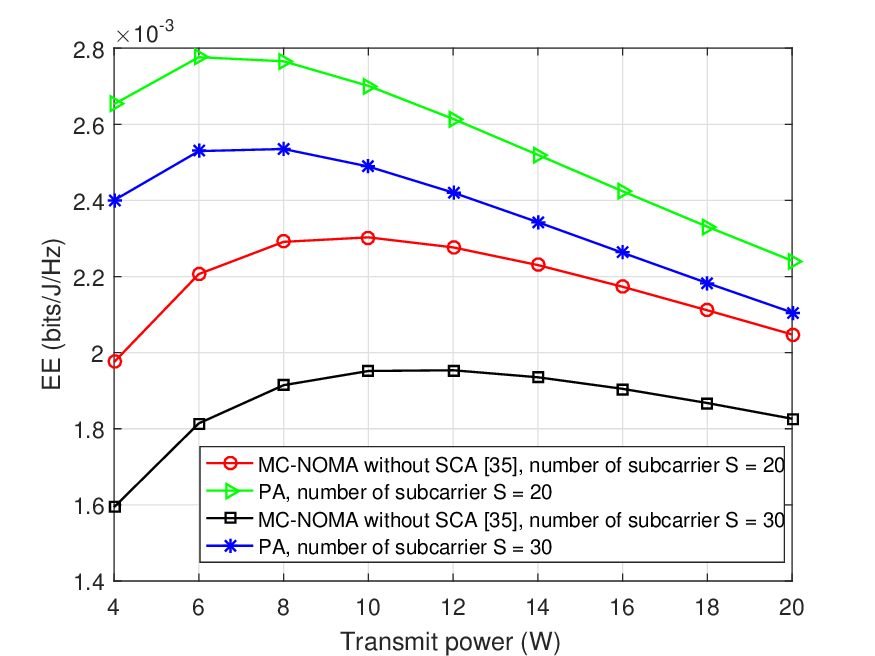}\vspace{-0.5em}
\caption{EE versus transmit power under different subcarrier schemes.}
\end{minipage} \vspace{-0.5em}
\end{figure*}

Fig. 8 shows the EE performance versus transmit power under the different subcarrier processing mode. For simplicity, we use ``MC-NOMA without SCA'' to denote the benchmark algorithm in the multi-carrier non-orthogonal multiple access system without considering the subcarrier allocation (SCA) \cite{8680645}.
Under imperfect CSI, we present the EE results of ``PA'' and ``MC-NOMA without SCA''. We can observe that the EE of ``PA'' has a better performance than that of ``MC-NOMA without SCA'' whether the number of subcarriers equals 20 or 30.
Moreover, the figure also shows that using more subcarriers reduces system EE.
When the number of subcarriers is 20, the EE of ``PA'' is higher than the case where 30 subcarriers are considered.
This is because more subcarriers narrow the width of the subcarrier bandwidths.
As a result, it requires more computations when allocating
the subcarriers to the IoT devices, reducing the system EE.

Considering the feature of massive MIMO, Fig. 9 shows the EE performance comparison between MIMO and massive MIMO.
We use the proposed AS-based distributed algorithm to solve the EE maximization.
The base station with 256/128/64 antenna arrays serves 15 single-antenna IoT devices.
We also compare the EE of the traditional MIMO network when the number of antennas is 2, 4, and 8, with the number of massive MIMO is 64, 128, and 256. We see that the EE with the massive MIMO-NOMA technique outperforms than that with the traditional MIMO-NOMA technique whether in Fig. 9(a) or (b). As the total antenna number $M_k$ increases, the increasing rate of system EE slows down. This can be seen from the massive MIMO case corresponding to $M_k = 64, 128, 256$. As we can observe, the EE with massive MIMO-NOMA technique under the case of $M_k = 256$ approaches the EE under the case of $M_k\rightarrow\infty$. These results verify that massive MIMO is superior to traditional MIMO due to its multi-antenna advantages, and infinite EE cannot be achieved by taking the total number of antennas to infinity. Fig. 9 also verifies that multi-antenna selection in massive MIMO-NOMA systems can obtain a high EE.

In Fig. 10, we evaluate the EE by jointly varying the transmit power and WPT time.
As shown in Fig. 10, the optimal WPT time slot $\tau_{k}$ can be obtained for a given transmit power.
For different numbers of selected antennas, Fig. 10 presents the EE of four cases,
Case 1: $N_{k, u, s}$ = 20; Case 2: $N_{k, u, s}$ = 25; Case 3: $N_{k, u, s}$ = 30; and Case 4: $N_{k, u, s}$ = 40.
We observe that, the EE firstly increases, then climbs to the maximum and decreases with the increasing WPT time.
When the WPT time slot is between $0$ and $1$ (e.g., 0.5), the EE first ascends, then arrives at the maximum and finally descends.
From the four subfigures of Fig. 10, we see that the system EE varies with the different antenna number. For example, the maximal EE with $N_{k, u, s} = 30 $ antennas is $6 \times 10^{-4}$ bits/J/Hz, higher than the other three cases.
Therefore, on top of AS, the joint optimization of the power $P_k$ and the WPT time $\tau_{k}$ is significant to obtain the maximum EE.
Fig. 10 confirms the validity of the proposed energy-efficient resource allocation algorithm in which the power and time are jointly optimized.

\begin{figure*}
\begin{minipage}[t]{0.5\linewidth}
\includegraphics[width=8.4cm]{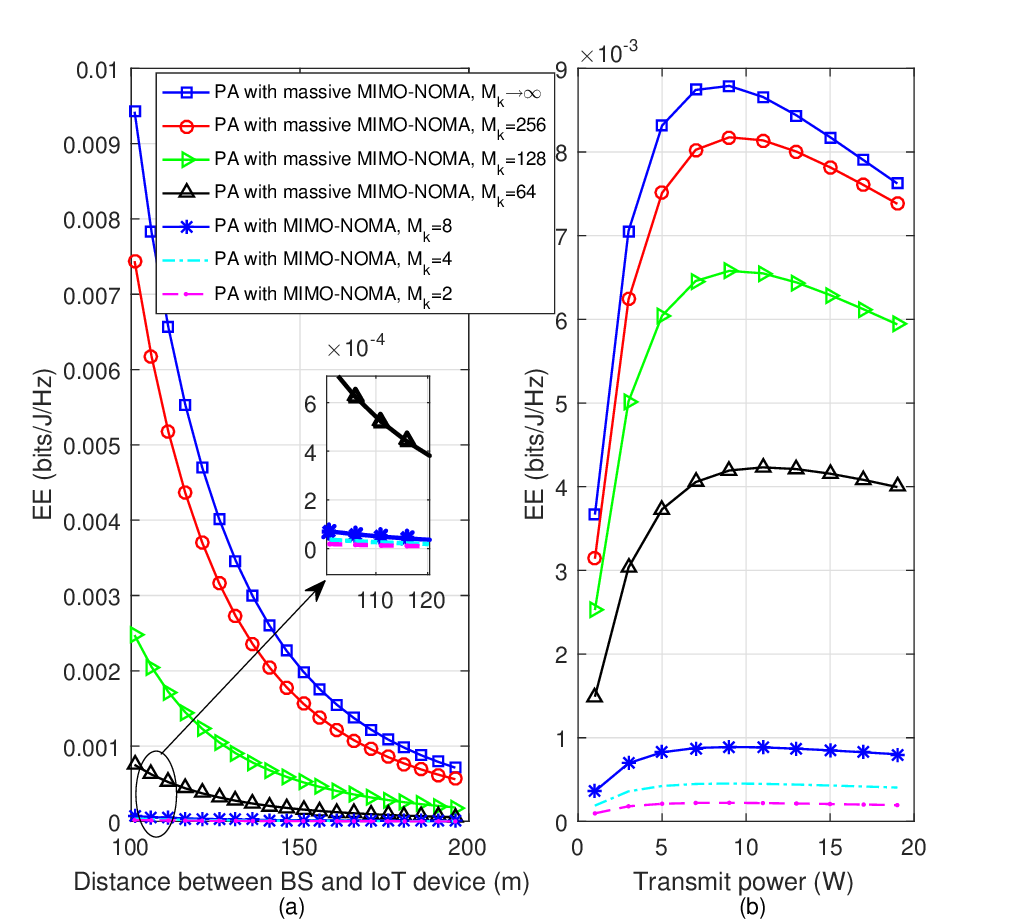} \vspace{-0.5em}
\caption{EE comparison between massive MIMO-NOMA and traditional MIMO-NOMA. Fig. 9(a) presents the EE versus transmission distance under different antenna number; Fig. 9(b) presents the EE versus transmit power under different antenna number.}
\end{minipage} \vspace{-1.5em}
\begin{minipage}[t]{0.5\linewidth}
\includegraphics[width=8.4cm ]{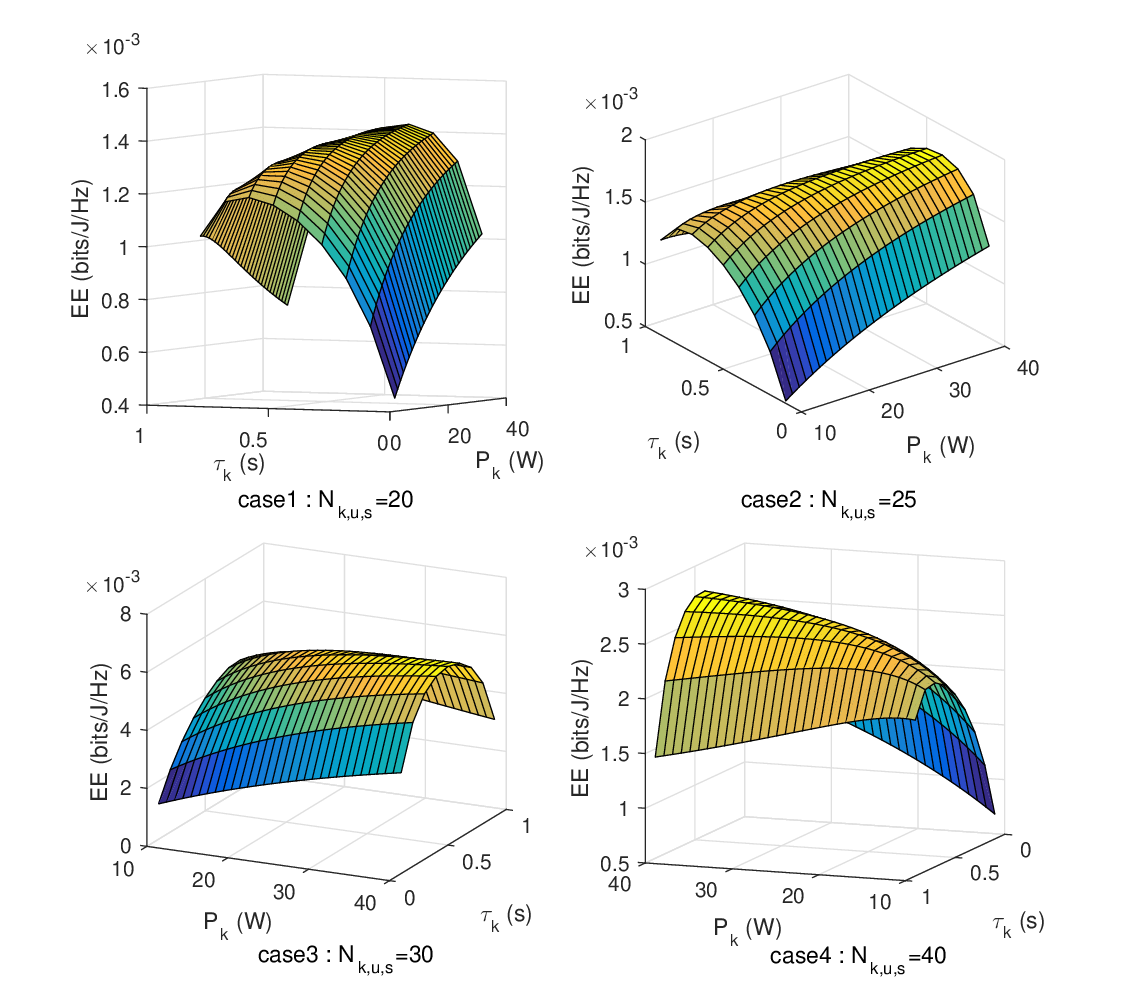} \vspace{-0.5em}
\caption{EE versus transmit power and first time slot.} \vspace{-1.5em}
\end{minipage}
\end{figure*}


To show the convergence performance of ``PA'', we compare it with a Dinkelbach-based centralized algorithm \cite{7335602}, and an OMA-based scheme using the Lagrange multipliers \cite{7332956}, in Fig. 11.
We denote the benchmark algorithms as ``Dinkelbach-based CA'' \cite{7582543} and ``LA OMA'' \cite{7332956}.
The convergence of ``PA'' and
the impact of the convergence parameter $\rho$ are presented in Fig. 11.
It shows that the optimal EE with a small number of
iterations can be achieved by adopting ADMM.
When $\rho$ = 0.068, the optimal EE can be obtained within nine iterations.
When $\rho$ = 0.088, the optimal EE can be obtained within seven iterations.
These results illustrate that the convergence rate is dependent on the value of $\rho$.
We also analyse the convergence difference of among ``Dinkelbach-based CA'', ``LA OMA'' and ``PA'' to show the performance loss in terms of the EE and convergence rate.
Comparing ``PA'' with ``Dinkelbach-based CA'' and ``LA OMA'', we find that the EE first increases as the  iteration index grows, and then the system EE stays unchanged. It can also be found that ``PA'' and ``Dinkelbach-based CA'' have very similar EE performances.
Fig. 11 shows that the EE performance difference is negligible, and the distributed method has fast convergence.
This is due to the vibration effect of the Dinkelbach method in the vicinity of the optimal solution. We remark that the Dinkelbach method is merely on the basis of closed-form expressions which require low computational complexity per
iteration but more iterations. The reason for ``PA'' with a higher convergence rate is that ``PA'' splits the target function of the primal variable to simplify the optimization problem and the augmented Lagrangian function has a saddle point contributing a faster convergence.
\begin{figure}
\includegraphics[width=8.4cm]{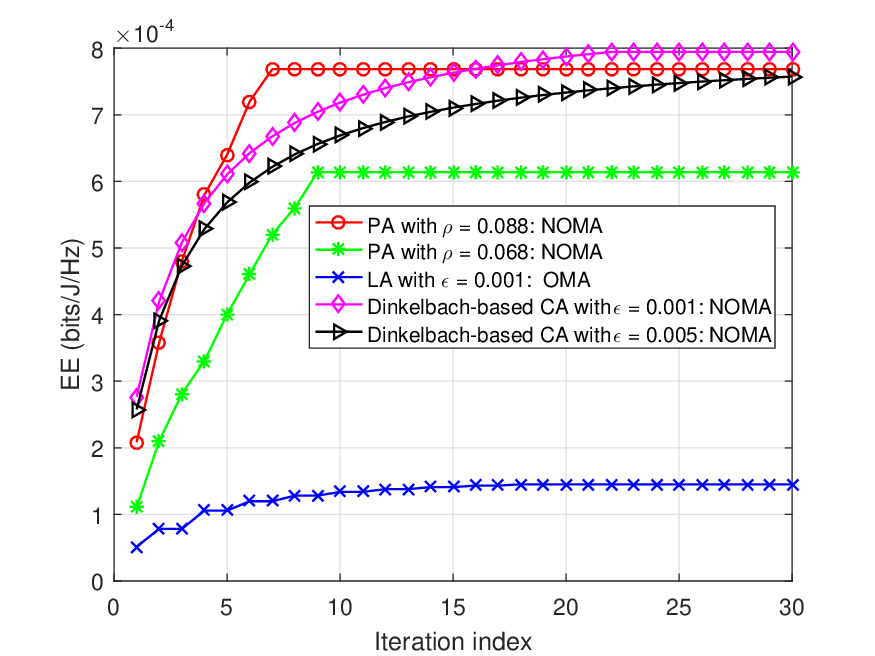}\vspace{-0.5em}
\caption{Convergence of the proposed algorithm and reference algorithms.}
 \vspace{-1em}
\end{figure}

\section{Conclusion}
This paper has investigated the energy-efficient resource allocation in the WPT-based massive MIMO-NOMA networks.
The joint transmit power, WPT time, antenna selection and subcarrier allocation scheme has been proposed to solve the system EE maximization problem. We have invoked the nonlinear fraction programming approach to convert the original non-convex problem to be convex, and solve the problem by developing a novel distributed ADMM-based energy-efficient resource allocation algorithm.
We have considered both the perfect and imperfect CSI, and analysed the impact of channel estimation error on the EE performance.
Simulation results have shown the convergence and effectiveness of the proposed algorithm, as well as the superior EE performance over the benchmark schemes.

There are many promising directions worth researching in the future. For instance, the IRS-assisted wireless communication is a young-born paradigm that has attracted researchers' attention. We can consider IRS-enabled massive MIMO-NOMA to improve the EE of the RF-based WPT networks.
UAV-assisted massive MIMO-NOMA is also an interesting topic, as it can effectively support the coverage and throughput of wireless communication by leveraging the strength of line-of-sight connections.
Moreover, it is worth extending our work to more practical and complicated scenarios, such as the networks with multiple dynamic IoT devices using different energy harvesting techniques.
\begin{appendices}

\section{Derivation of (17)}
In the considered WPT enabled massive MIMO-NOMA network, we take the $k$-th cell as a example, where an IoT device transmits data packets on the $s$-th subcarrier.
For simplicity, we omit the subscript $k$ and $s$.  According to \cite{Stigler}, we have ${{x}_{1}}>{{x}_{2}}>\cdots >{{x}_{u}}\cdots >{{x}_{U}}$ as the ordered random signal variables from the IoT devices in the $k$-th cell. With massive antennas at each BS,
the selected optimal antenna number $N$ satisfies the condition of $1\le N\le M$. When $M \rightarrow \infty$, the distribution of trimmed sum $\sum\limits_{u=1}^{N}{{{x}_{u}}}$ of IoT devices is asymptotically normal.
When ${{x}_{u}}$ is a chi-square random variable signal propagated through ${\hat{h}}_{u}$ with the degree of freedom 2, the mean and variance can be found in \cite{5766166}. The channel coefficient sum between the $u$-th IoT device and the $BS_k$
can be denoted by
\begin{equation}
\label{add1218-01}
\sum\limits_{u=1}^{N}{{{\left| {{{\hat{h}}}_{u}} \right|}^{2}}} \sim \mathcal{N}\left( N\left( 1+\ln \frac{M}{N} \right),N\left( 2-\frac{N}{M} \right) \right)
\end{equation}
where ${{\hat{h}}_{u}}$ denotes the subchannel of channel matrix $\mathbf{\hat{H}}$ between the IoT devices and the BS. According to the \cite{7906591}, in our proposed WPT enabled massive MIMO-NOMA network,
the mutual information under the antenna selection scheme is expressed by
\begin{small}
\begin{equation}
\label{add1218-02}
{{I}_{sel}}\text{=}{{\log }_{2}}\left| 1\text{+}\gamma \sum\limits_{u=1}^{N}{{{\left| {{{\hat{h}}}_{u}} \right|}^{\text{2}}}} \right|
\text{=}{{\log }_{2}}\left| 1\text{+(1+ln}\frac{M}{N}\text{)}\gamma N \right|+{{\log }_{2}}\left| \xi \right|
\end{equation}
\end{small}

\vspace{-1em}
\noindent where $\gamma $ is the SNR measured at each receive antenna. $\left| \cdot  \right|$ takes the absolute value, and the introduced new variable $\xi$ is given by
\begin{equation}
\label{add1218-03}
\xi=1+\frac{\gamma \left( \sum\limits_{u=1}^{N}{{{\left| {{{\hat{h}}}_{u}} \right|}^{2}}}-N\left( 1+\ln \frac{M}{N} \right) \right)}{1+\left( 1+\ln \frac{M}{N} \right)\gamma N}.
\end{equation}
According to (\ref{add1218-01}), $\xi$ obeys the distribution as follows
\begin{equation}
\label{add1218-04}
\xi \sim \mathcal{N}\left( 1,\frac{{{\gamma }^{2}}N\left( 2-\frac{N}{M} \right)}{{{\left( 1+\left( 1+\ln \frac{M}{N} \right)\gamma N \right)}^{2}}} \right).
\end{equation}
As the property of folded normal distribution of variable $y\text{=}\left| \xi \right|$, the mutual information of WPT enabled massive MIMO-NOMA networks can be further derived as
\begin{equation}
\begin{aligned}
\label{add1218-08}
&{{I}_{sel}}\text{=}{{\log }_{2}}\left(1\text{+}(1\text{+}\ln \frac{M}{N} )\gamma N \right)\text{+}(y-1){{\log }_{2}}e\text{+}O\left( {{(y-1)}^{2}} \right).
\end{aligned}
\end{equation}
As $O\left( {{(y-1)}^{2}} \right)$ is asymptotically zero \cite{7906591}, ${{I}_{sel}}$ obeys the distribution as follows
\begin{small}
\begin{equation}
\begin{aligned}
\label{add1218-09}
{{I}_{sel}} \sim \mathcal{F}\mathcal{N}\left( {{\log }_{2}}[ 1+( 1+\ln \frac{M}{N} )\gamma N ] \right.,\left. \frac{{{\left( {{\log }_{2}}e \right)}^{2}}{{\gamma }^{2}}N\left( 2-\frac{N}{M} \right)}{{{\left( 1+\left( 1+\ln \frac{M}{N} \right)\gamma N \right)}^{2}}} \right).
\end{aligned}
\end{equation}
\end{small}

The proof of data rate in (\ref{equa8}) under imperfect CSI is completed. \ \ \ \ \ \ \ \ \ \ \ \ \ \ \ \ \ \ \ \ \ \ \ \ \ \ \ \ \ \ \ \ \ \ \ \ \ \ \ \ \ \ \ \ \ \ \ \ \ \ \ \ \ \ \ \ \ \ \  \qedsymbol

\section{The Concavity Proof of $R_{tot}\left(\boldsymbol{P}, \boldsymbol{\tau}, \boldsymbol{N}, \boldsymbol{C}\right)$ in (\ref{equa18})}

To prove the joint concavity of $R_\text{tot}\left(\boldsymbol{P}, \boldsymbol{\tau}, \boldsymbol{N}, \boldsymbol{C}\right)$ in (\ref{equa18}) concerning the variables $\boldsymbol{P},\boldsymbol{\tau},\boldsymbol{L}$ and $\boldsymbol{C}$, we first take successive convex approximation methods to relax the discrete variables $N_{k, u, s}$ and $c_{k, u, s}$ into continuous slack variables. Further, the Hessian matrix of $R_{tot}$ can be denoted by

\begin{equation}
\begin{aligned}
&{{\mathbf{H}}_{R}}=\\
&\left[ \begin{matrix}
   \frac{{{\partial }^{2}}{{R}_{\text{tot}}}}{\partial P_{k}^{2}} & \frac{{{\partial }^{2}}{{R}_{\text{tot}}}}{\partial {{P}_{k}}\partial {{\tau }_{k}}} & \frac{{{\partial }^{2}}{{R}_{\text{tot}}}}{\partial {{P}_{k}}\partial N_{k, u, s}} & \frac{{{\partial }^{2}}{{R}_{\text{tot}}}}{\partial {{P}_{k}}\partial {c_{k, u, s}}}  \\
   \frac{{{\partial }^{2}}{{R}_{\text{tot}}}}{\partial {{\tau }_{k}}\partial {{P}_{k}}} & \frac{{{\partial }^{2}}{{R}_{\text{tot}}}}{\partial \tau _{k}^{2}} & \frac{{{\partial }^{2}}{{R}_{\text{tot}}}}{\partial {{\tau }_{k}}\partial N_{k, u, s}} & \frac{{{\partial }^{2}}{{R}_{\text{tot}}}}{\partial {{\tau }_{k}}\partial {c_{k, u, s}}}  \\
   \frac{{{\partial }^{2}}{{R}_{\text{tot}}}}{\partial N_{k, u, s}\partial {{P}_{k}}} & \frac{{{\partial }^{2}}{{R}_{\text{tot}}}}{\partial N_{k, u, s}\partial {{\tau }_{k}}} & \frac{{{\partial }^{2}}{{R}_{\text{tot}}}}{\partial N_{k, u, s}^{2}} & \frac{{{\partial }^{2}}{{R}_{\text{tot}}}}{\partial N_{k, u, s}\partial {c_{k, u, s}}}  \\
   \frac{{{\partial }^{2}}{{R}_{\text{tot}}}}{\partial {c_{k, u, s}}\partial {{P}_{k}}} & \frac{{{\partial }^{2}}{{R}_{\text{tot}}}}{\partial {c_{k, u, s}}\partial {{\tau }_{k}}} & \frac{{{\partial }^{2}}{{R}_{\text{tot}}}}{\partial {c_{k, u, s}}\partial N_{k, u, s}} & \frac{{{\partial }^{2}}{{R}_{\text{tot}}}}{\partial c_{k, u, s}^{2}}  \\
\end{matrix} \right].
\end{aligned}
\end{equation}

More specifically, we calculate the partial derivatives of each element of the Hessian matrix and we have
$\frac{{{\partial }^{2}}{{R}_{\text{tot}}}}{\partial P_{k}^{2}}\text{=-}\frac{f_{11}^{1}f_{11}^{3}\left( 2{{\left( f_{11}^{2} \right)}^{2}}{{P}_{k}}+2f_{11}^{1}f_{11}^{2}{{P}_{k}}+2f_{11}^{2}f_{11}^{3}+f_{11}^{1}f_{11}^{3} \right)}{{{\left( \left[ {{\left( f_{11}^{2} \right)}^{2}}+f_{11}^{1}f_{11}^{2} \right]P_{k}^{2}+\left( 2f_{11}^{2}f_{11}^{3}+f_{11}^{1}f_{11}^{3} \right){{P}_{k}}+{{\left( f_{11}^{3} \right)}^{2}} \right)}^{2}}}<0$, \\where
$f_{11}^{1}\text{=}KUS{{B}_s}\eta {{\tau }_{k}}\alpha _{k,u}^{4}{{\left| {{\mathbf{h}}_{k, u, s}} \right|}^{4}}$,
$f_{11}^{2}\text{=}\left( U-u+1 \right)\eta {{\tau }_{k}}\cdot$ $\alpha _{k,{{u}_{1}}}^{4} {{\left| {{\mathbf{h}}_{k,u_1,s}} \right|}^{4}}{{c}_{k,u_1,s}}$,
$f_{11}^{3}\text{=}\left( T\text{-}{{\tau }_{k}} \right)\left( {{I}_{k_2,u,s}}+{{\sigma }^{2}} \right)$. Similarly, we calculate the second partial derivatives of the other elements of the upper triangular Hessian matrix, and obtain that the sign of each element is less than 0 according to the property of symmetric matrix. We take any two independent variables ${{x}_{1}}$ and ${{x}_{2}}$ (${{x}_{1}},{{x}_{2}}\in \{\boldsymbol{P}, \boldsymbol{\tau}, \boldsymbol{N}, \boldsymbol{C}\}$), and the difference between ${{x}_{1}}$ and ${{x}_{2}}$  is denoted by $\Delta \mathbf{x}$, then the inequality $\Delta {{\mathbf{x}}^{T}}\mathbf{H}_R\Delta \mathbf{x}\le 0$ holds. Thus, ${{\mathbf{H}}_{R}}$ is semi-negative definite, and the concavity of ${{R}_{\text{tot}}}$ is proved.

The concavity proof of (\ref{equa18}) is completed. \ \ \ \ \ \ \ \ \ \ \ \ \ \ \ \ \ \qedsymbol

\section{Derivation of Theorem 1}
\subsection{Sufficiency}
The proof of Theorem 1 can be established in the same way as in \cite{WD}. For simplifying notations, we define $\Delta $, $\Theta $, $\Re$ and $\Im$ as the set of feasible solutions for $\mathbf{P1}$ in (\ref{equP1}), and
${\eta_{\text{EE}}\left(\boldsymbol{P}, \boldsymbol{\tau}, \boldsymbol{N}, \boldsymbol{C}\right)=}
{R_\text{tot}\left(\boldsymbol{P}, \boldsymbol{\tau}, \boldsymbol{N}, \boldsymbol{C}\right)}/{E_\text{tot}\left(\boldsymbol{P}, \boldsymbol{\tau}, \boldsymbol{N}, \boldsymbol{C}\right)}$.
Without loss of generality,  $\eta_{\text{EE}}^*$, $P_k^* \in \Delta$, $\tau_{k}^* \in \Theta$, $N_{k, u, s}^* \in \Re$ and $c_{k, u, s}^* \in \Im$ are defined as the optimal EE, power allocation, time partition scheme, AS and subcarrier allocation policies of the primal problem $\mathbf{P1}$, respectively. We can obtain the optimal EE as
\begin{equation}
\label{appendix-suff1}
\begin{aligned}
&\eta_{\text{EE}}^* = \frac{R_\text{tot}\left(\boldsymbol{P}^*,\boldsymbol{\tau}^*,\boldsymbol{N}^*, \boldsymbol{C}^* \right)}{E_\text{tot}\left(\boldsymbol{P}^*,\boldsymbol{\tau}^*,\boldsymbol{N}^*, \boldsymbol{C}^*\right)}
 \ge \frac{R_\text{tot}\left(\boldsymbol{P}, \boldsymbol{\tau}, \boldsymbol{N}, \boldsymbol{C}\right)}{E_\text{tot}\left(\boldsymbol{P}, \boldsymbol{\tau}, \boldsymbol{N}, \boldsymbol{C}\right)}, \\
 & \ \ \ \ \ \forall \ P_k \in \Delta, \tau_{k} \in \Theta, N_{k, u, s} \in \Re, c_{k, u, s} \in \Im.\\
\end{aligned}
\end{equation}

Since the denominator of (\ref{appendix-suff1}) is positive, according to the transformation of the equation and inequality, we have
\begin{equation}
\left\{ \begin{aligned}
  & {{R}_{\text{tot}}}\left( \boldsymbol{P}, \boldsymbol{\tau}, \boldsymbol{N}, \boldsymbol{C} \right)-\eta _{EE}^{*}{{E}_{\text{tot}}}\left( \boldsymbol{P}, \boldsymbol{\tau}, \boldsymbol{N}, \boldsymbol{C} \right)\le 0, \\
 & {{R}_{\text{tot}}}\left( \boldsymbol{P}^*,\boldsymbol{\tau}^*,\boldsymbol{N}^*, \boldsymbol{C}^* \right)-\eta _{EE}^{*}{{E}_{\text{tot}}}\left( \boldsymbol{P}^*,\boldsymbol{\tau}^*,\boldsymbol{N}^*, \boldsymbol{C}^* \right)=0 .\\
\end{aligned} \right.
\end{equation}
We conclude that $\underset{\boldsymbol{P},\boldsymbol{\tau},\boldsymbol{N},\boldsymbol{C}}{\mathop{\max }}  R_\text{tot}\left( \boldsymbol{P}, \boldsymbol{\tau}, \boldsymbol{N}, \boldsymbol{C}  \right)-\eta_\text{EE}$ $\cdot E_\text{tot}\left(\boldsymbol{P}, \boldsymbol{\tau}, \boldsymbol{N}, \boldsymbol{C} \right)$ $ = 0$ is achievable and feasible at the optimal power, time, AS and subcarrier allocation $\boldsymbol{P}^*$, $\boldsymbol{\tau}^* $, $\boldsymbol{N}^*$, $\boldsymbol{C}^*$.

\subsection{Necessity}
The converse implication of Theorem 1 can be proved as follows. $\boldsymbol{P}'$, $\boldsymbol{\tau}'$, $\boldsymbol{N}'$ and $\boldsymbol{C}'$ are assumed as the optimal power, WPT time, AS and subcarrier allocation policies of the equivalent objective function so as to satisfy the following condition
\begin{equation}
 \begin{aligned}
 & {{R}_{\text{tot}}}\left( \boldsymbol{P}',\boldsymbol{\tau}',\boldsymbol{N}',\boldsymbol{C}' \right)-\eta _{EE}^{*}{{E}_{\text{tot}}}\left( \boldsymbol{P}',\boldsymbol{\tau}',\boldsymbol{N}',\boldsymbol{C}' \right)=0 \\
\end{aligned} .
\end{equation}

Given any feasible power, time, subcarrier allocation and AS policies, $P_k \in \Delta$, $\tau_{k} \in \Theta$, $c_{k, u, s} \in \Im$ and $N_{k, u, s} \in \Re$, we have
\begin{equation}
\label{ineq2}
 \begin{aligned}
 & \ \ \ \ {{R}_{\text{tot}}}\left( \boldsymbol{P},\boldsymbol{\tau},\boldsymbol{N},\boldsymbol{C} \right)-\eta _{EE}^{*}{{E}_{\text{tot}}}\left( \boldsymbol{P},\boldsymbol{\tau},\boldsymbol{N},\boldsymbol{C} \right) \\
&  \leq
  {{R}_{\text{tot}}}\left( \boldsymbol{P}',\boldsymbol{\tau}',\boldsymbol{N}',\boldsymbol{C}' \right)-\eta _{EE}^{*}{{E}_{\text{tot}}}\left( \boldsymbol{P}',\boldsymbol{\tau}',\boldsymbol{N}',\boldsymbol{C}'\right)=0.
\end{aligned}
\end{equation}
The inequality (\ref{ineq2}) yields
\begin{align}
& \frac{R_\text{tot}\left(\boldsymbol{P},\boldsymbol{\tau},\boldsymbol{N},\boldsymbol{C}\right)}
{E_\text{tot}\left(\boldsymbol{P},\boldsymbol{\tau},\boldsymbol{N},\boldsymbol{C}\right)} \leq \eta_{\text{EE}}^*, \nonumber\\
&\forall \ P_k \in \Delta, \tau_{k} \in \Theta, N_{k, u, s} \in \Re, c_{k, u, s} \in \Im\\
& \frac{R_\text{tot}\left(\boldsymbol{P}',\boldsymbol{\tau}',\boldsymbol{N}',\boldsymbol{C}'
\right)}{E_\text{tot}\left(\boldsymbol{P}',\boldsymbol{\tau}',\boldsymbol{N}',\boldsymbol{C}'\right)} = \eta_{\text{EE}}^* .
\end{align}
This indicates that the optimal solutions $\boldsymbol{P}'$, $\boldsymbol{\tau}'$, $\boldsymbol{N}'$ and $\boldsymbol{C}'$ for the objective function corresponds to the optimal resource allocation policies for the original objective $\mathbf{P1}$.

The proof of Theorem 1 is completed. \ \ \ \ \ \ \ \ \ \ \ \ \ \ \ \ \ \qedsymbol

\section{The Convexity Proof of P4}
In light of the definition of Hessian matrix, if $f$ is a real-valued function $f(x_1,x_2,...,x_n)$ and differentiable with respect to the independent variable $x_i, i\in [1,n]$, the Hessian matrix of $f$ is defined as $\mathbf{H}(f)_{i,j}(x)=D_i D_j f(x)$, where $D_i$ (and $D_j$) denotes the differential operator with respect to the $i$-th ($j$-th) independent variable. The Hessian matrix of $\Lambda( {{{P}_{k}}},{\tilde{{\tau }}_{k}},{{\tilde{n}}_{k, u, s}},{\tilde{c}_{k, u, s}} )$ of $\mathbf{P4}$ in (\ref{equp4}) can be given by
\begin{equation}
\label{HessianMat}
\mathbf{H}=\left[ \begin{matrix}
   \frac{{{\partial }^{2}}{\Lambda}}{\partial P_{k}^{2}} & \frac{{{\partial }^{2}}{\Lambda}}{\partial {{P}_{k}}\partial {{{\tilde{\tau }}}_{k}}} & \frac{{{\partial }^{2}}{\Lambda}}{\partial {{P}_{k}}\partial \tilde{n}_{k, u, s}}  \\
   \frac{{{\partial }^{2}}{\Lambda}}{\partial {{{\tilde{\tau }}}_{k}}\partial {{P}_{k}}} & \frac{{{\partial }^{2}}{\Lambda}}{\partial \tilde{\tau }_{k}^{2}} & \frac{{{\partial }^{2}}{\Lambda}}{\partial {{{\tilde{\tau }}}_{k}}\partial \tilde{n}_{k, u, s}}  \\
   \frac{{{\partial }^{2}}{\Lambda}}{\partial \tilde{n}_{k, u, s}\partial {{P}_{k}}} & \frac{{{\partial }^{2}}{\Lambda}}{\partial \tilde{n}_{k, u, s}\partial {{{\tilde{\tau }}}_{k}}} & \frac{{{\partial }^{2}}{\Lambda}}{\partial \tilde{n}_{k, u, s}^{2}}  \\
\end{matrix} \right].
\end{equation}
In (\ref{HessianMat}), we let $\Lambda$ = $\Lambda( {{{P}_{k}}},{\tilde{{\tau }}_{k}},{{\tilde{n}}_{k, u, s}},{\tilde{c}_{k, u, s}})$ for notational simplicity. Further, we divide $\mathbf{H}$ into the sub-matrices according to the order of the matrix, as given by
\begin{equation}
\mathbf{H_1} = \left[ \begin{matrix}
   \frac{{{\partial }^{2}}{\Lambda}}{\partial P_{k}^{2}}\end{matrix} \right],
   \mathbf{H_2} = \left[ \begin{matrix}
   \frac{{{\partial }^{2}}{\Lambda}}{\partial P_{k}^{2}} & \frac{{{\partial }^{2}}{\Lambda}}{\partial {{P}_{k}}\partial {{{\tilde{\tau }}}_{k}}}  \\
   \frac{{{\partial }^{2}}{\Lambda}}{\partial {{{\tilde{\tau }}}_{k}}\partial {{P}_{k}}} & \frac{{{\partial }^{2}}{\Lambda}}{\partial \tilde{\tau }_{k}^{2}}  \\
   \end{matrix} \right],
   \mathbf{H_3} = \mathbf{H}.
\end{equation}
As the determinants of sub-matrices affect the positive definition of the original matrix, by evaluating the determinants of the sub-matrices, we have
\begin{equation}
\label{determinant}
   | \mathbf{H_1}| \leq 0,
   | \mathbf{H_2}| \geq 0,
   | \mathbf{H_3}| \leq 0.
\end{equation}
Thus, $\mathbf{H}$ is semi-negative definite. We can conclude that $\Lambda( {{{P}_{k}}},{\tilde{{\tau }}_{k}},$ ${{\tilde{n}}_{k, u, s}},{\tilde{c}_{k, u, s}} )$  of $\mathbf{P4}$ is a jointly convex function of variables ${P}_{k}, {\tilde{{\tau }}_{k}}$, and ${{\tilde{n}}_{k, u, s}}$. 

The convexity proof of $\mathbf{P4}$ is completed. \ \ \ \ \ \ \ \ \ \ \ \ \ \ \ \ \ \qedsymbol
\end{appendices}

\begin{IEEEbiography}[{\includegraphics[width=1in,height =1.25in,clip,keepaspectratio]{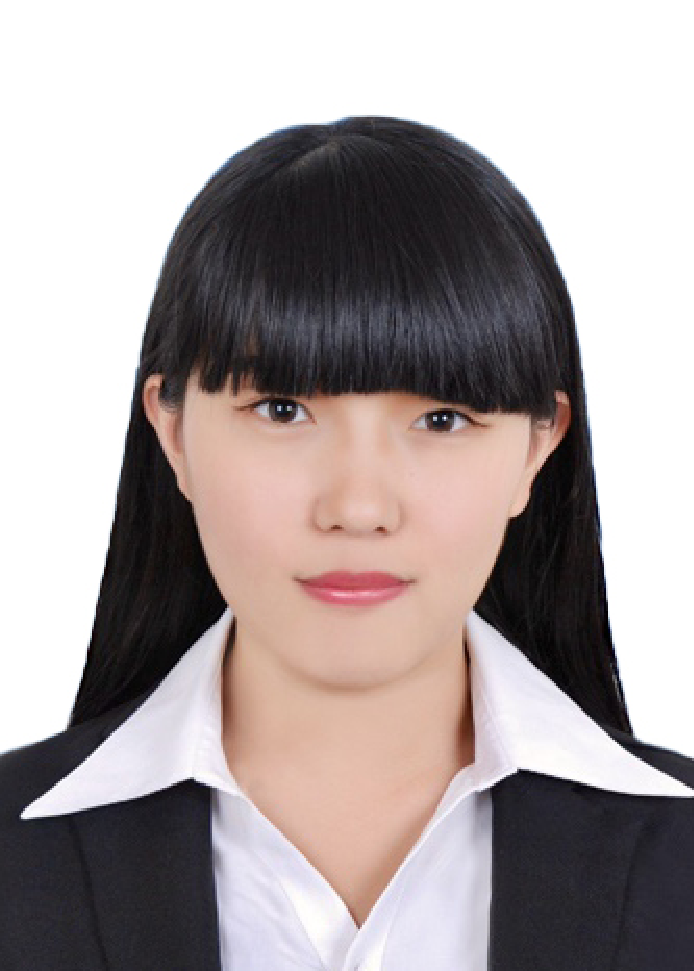}}]{Zhongyu Wang} (S'19) received the M.S. degree in computer science and technology from School of Information Science and Engineering, Yanshan University, Qinhuangdao, China. She is currently working toward the Ph.D. degree in information and communication engineering with the School of Information and Communication Engineering, Beijing University of Posts and Telecommunications, Beijing, China. Her current research interests include ultra-reliable low latency communication, non-orthogonal multiple access, massive MIMO, energy harvesting technology, and UAV communication.
\end{IEEEbiography}
\begin{IEEEbiography}[{\includegraphics[width=1in,height =1.25in,clip,keepaspectratio]{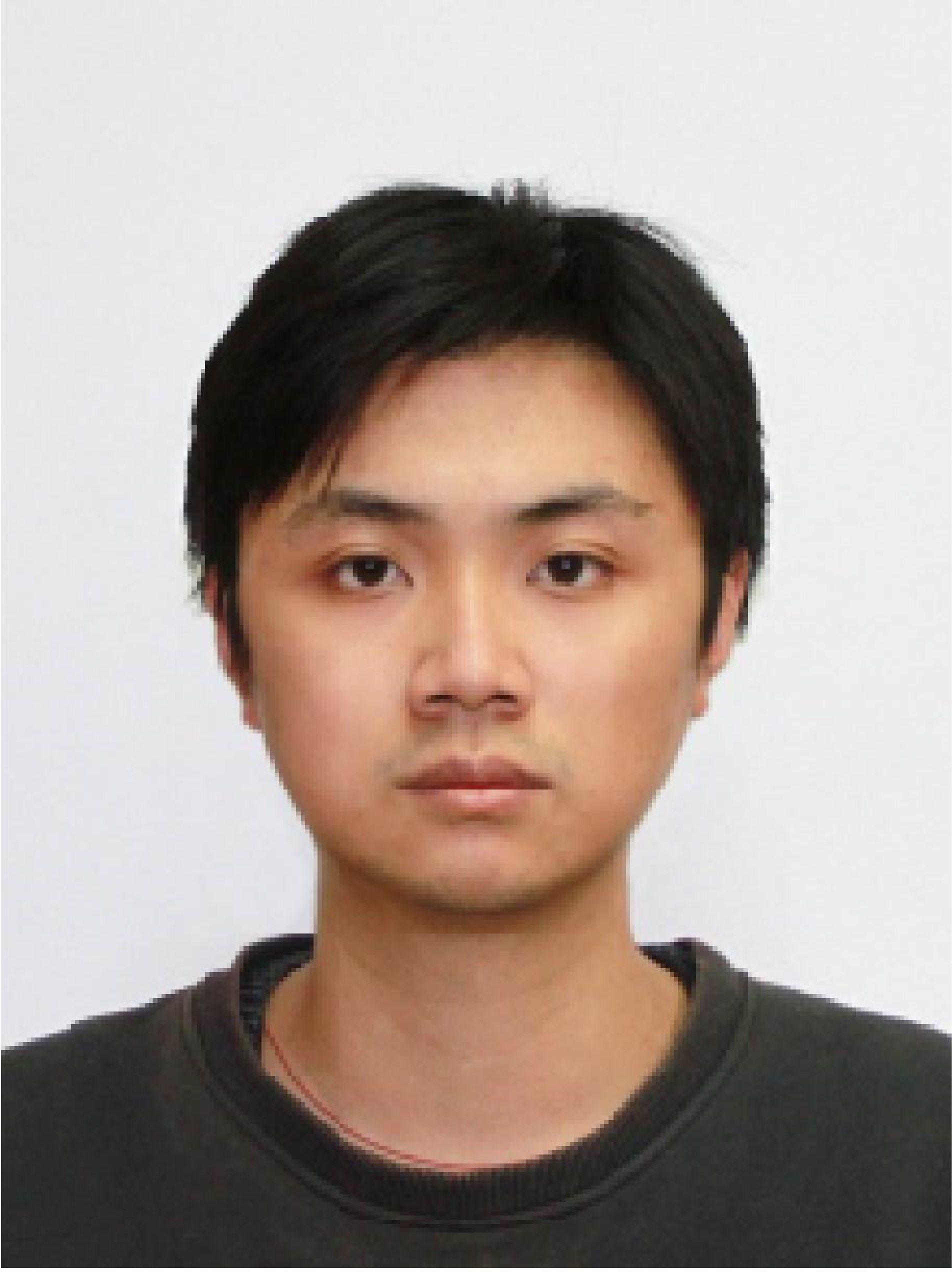}}]{Zhipeng Lin} (M'20) is currently working toward the dual Ph.D. degrees in communication and information engineering with the School of Information and Communication Engineering, Beijing University of Posts and Telecommunications, Beijing, China, and the School of Electrical and Data Engineering, University of Technology of Sydney, NSW, Australia. His current research interests include millimeter-wave communication, massive MIMO, hybrid beamforming, wireless localization, and tensor processing.
\end{IEEEbiography}
\begin{IEEEbiography}[{\includegraphics[width=1in,height
=1.25in,clip,keepaspectratio]{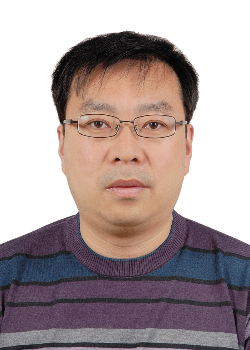}}]{Tiejun Lv}
 (M'08-SM'12) received the M.S. and Ph.D. degrees in electronic engineering from the University of Electronic Science and Technology of China (UESTC), Chengdu, China, in 1997 and 2000, respectively. From January 2001 to January 2003, he was a Postdoctoral Fellow with Tsinghua University, Beijing, China. In 2005, he was promoted to a Full Professor with the School of Information and Communication Engineering, Beijing University of Posts and Telecommunications (BUPT). From September 2008 to March 2009, he was a Visiting Professor with the Department of Electrical Engineering, Stanford University, Stanford, CA, USA. He is the author of 3 books, more than 90 published IEEE journal papers and 190 conference papers on the physical layer of wireless mobile communications. His current research interests include signal processing, communications theory and networking. He was the recipient of the Program for New Century Excellent Talents in University Award from the Ministry of Education, China, in 2006. He received the Nature Science Award in the Ministry of Education of China for the hierarchical cooperative communication theory and technologies in 2015.
\end{IEEEbiography}
\begin{IEEEbiography}[{\includegraphics[width=1in,height =1.25in,clip,keepaspectratio]{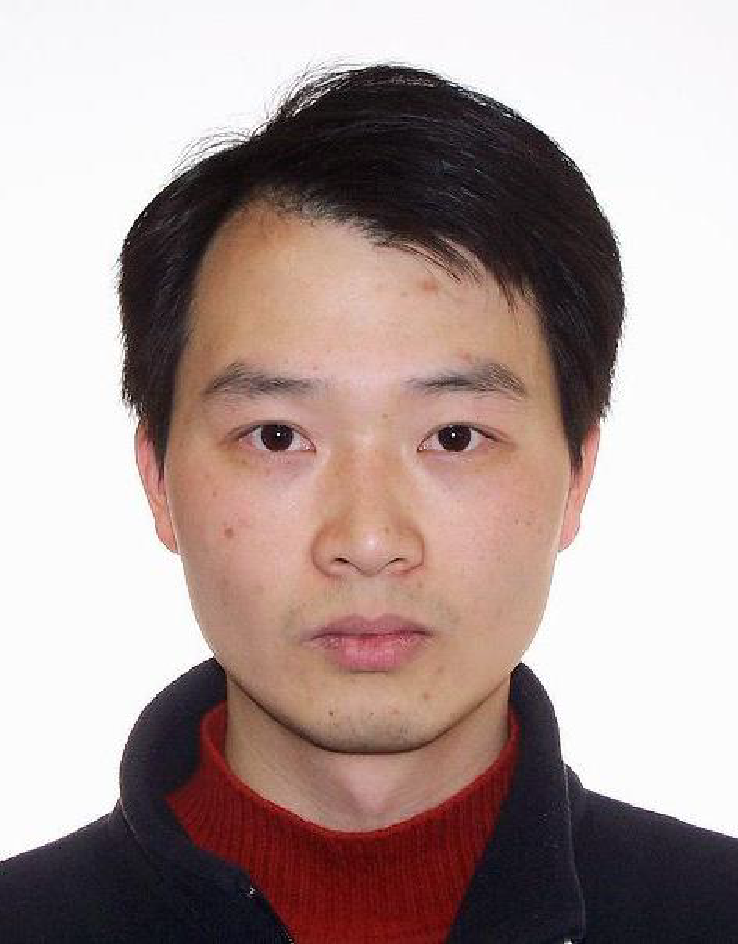}}]{Wei Ni} (M'09-SM'15) received the B.E. and Ph.D. degrees in Electronic Engineering from Fudan University, Shanghai, China, in 2000 and 2005, respectively. Currently, he is a Group Leader and Principal Research Scientist at CSIRO, Sydney, Australia, and an Adjunct Professor at the University of Technology Sydney and Honorary Professor at Macquarie University, Sydney. He was a Postdoctoral Research Fellow at Shanghai Jiaotong University from 2005 -- 2008; Deputy Project Manager at the Bell Labs, Alcatel/Alcatel-Lucent from 2005 to 2008; and Senior Researcher at Devices R\&D, Nokia from 2008 to 2009. His research interests include signal processing, optimization, learning, and their applications to network efficiency and integrity.

Dr Ni is the Chair of IEEE Vehicular Technology Society (VTS) New South Wales (NSW) Chapter since 2020 and an Editor of IEEE Transactions on Wireless Communications since 2018. He served first the Secretary and then Vice-Chair of IEEE NSW VTS Chapter from 2015 to 2019, Track Chair for VTC-Spring 2017, Track Co-chair for IEEE VTC-Spring 2016, Publication Chair for BodyNet 2015, and Student Travel Grant Chair for WPMC 2014.
\end{IEEEbiography}


\begin{thebibliography}{10}
\bibitem{8930983}
Q.~{Qi}, X.~{Chen}, and D.~W.~K. {Ng}, ``Robust beamforming for {NOMA}-based
  cellular massive {IoT} with {SWIPT},'' \emph{IEEE Trans. Signal Processing},
  vol.~68, pp. 211--224, 2020.

\bibitem{8534441}
F.~{Fang}, J.~{Cheng}, and Z.~{Ding}, ``Joint energy efficient subchannel and
  power optimization for a downlink {NOMA} heterogeneous network,'' \emph{IEEE
  Trans. Veh. Technol.}, vol.~68, no.~2, pp. 1351--1364, Feb. 2019.

\bibitem{8214104}
T.~D. {Ponnimbaduge Perera}, D.~N.~K. {Jayakody}, S.~K. {Sharma},
  S.~{Chatzinotas}, and J.~{Li}, ``Simultaneous wireless information and power
  transfer ({SWIPT}): Recent advances and future challenges,'' \emph{IEEE
  Commun. Surveys Tuts.}, vol.~20, no.~1, pp. 264--302, Firstquarter 2018.

\bibitem{7010878}
S.~{Ulukus}, A.~{Yener}, E.~{Erkip}, O.~{Simeone}, M.~{Zorzi}, P.~{Grover}, and
  K.~{Huang}, ``Energy harvesting wireless communications: A review of recent
  advances,'' \emph{IEEE J. Sel. Areas Commun.}, vol.~33, no.~3, pp. 360--381,
  Mar. 2015.

\bibitem{8485639}
L.~{Dai}, B.~{Wang}, M.~{Peng}, and S.~{Chen}, ``Hybrid precoding-based
  millimeter-wave massive {MIMO-NOMA} with simultaneous wireless information
  and power transfer,'' \emph{IEEE J. Sel. Areas Commun.}, vol.~37, no.~1, pp.
  131--141, Jan. 2019.

\bibitem{8444982}
T.~A. {Khan}, A.~{Yazdan}, and R.~W. {Heath}, ``Optimization of power transfer
  efficiency and energy efficiency for wireless-powered systems with massive
  {MIMO},'' \emph{IEEE Trans. Wireless Commun.}, vol.~17, no.~11, pp.
  7159--7172, Nov. 2018.

\bibitem{8629017}
L.~{Zhao} and X.~{Wang}, ``Massive {MIMO} downlink for wireless information and
  energy transfer with energy harvesting receivers,'' \emph{IEEE Trans.
  Commun.}, vol.~67, no.~5, pp. 3309--3322, May 2019.

\bibitem{6457363}
H.~Q. {Ngo}, E.~G. {Larsson}, and T.~L. {Marzetta}, ``Energy and spectral
  efficiency of very large multiuser {MIMO} systems,'' \emph{IEEE Trans.
  Commun.}, vol.~61, no.~4, pp. 1436--1449, Apr. 2013.

\bibitem{7062017}
X.~{Wang}, F.~{Zheng}, P.~{Zhu}, and X.~{You}, ``Energy-efficient resource
  allocation in coordinated downlink multicell {OFDMA} systems,'' \emph{IEEE
  Trans. Veh. Technol.}, vol.~65, no.~3, pp. 1395--1408, Mar. 2016.

\bibitem{8390925}
H.~{Zhang}, B.~{Wang}, C.~{Jiang}, K.~{Long}, A.~{Nallanathan}, V.~C.~M.
  {Leung}, and H.~V. {Poor}, ``Energy efficient dynamic resource optimization
  in {NOMA} system,'' \emph{IEEE Trans. Wireless Commun.}, vol.~17, no.~9, pp.
  5671--5683, Sep. 2018.

\bibitem{add20200411}
{L. P. {Qian} and Y. {Wu} and H. {Zhou} and X. {Shen}}, ``{Joint uplink base
  station association and power control for small-cell networks with
  non-orthogonal multiple access},'' \emph{{IEEE Trans. Wireless Commun.}},
  vol.~{16}, no.~{9}, pp. {5567--5582}, {Sep.} {2017}.

\bibitem{1291726}
N.~{Jindal}, S.~{Vishwanath}, and A.~{Goldsmith}, ``On the duality of gaussian
  multiple-access and broadcast channels,'' \emph{IEEE Trans. Inf. Theory},
  vol.~50, no.~5, pp. 768--783, May 2004.

\bibitem{8638930}
C.~{Xiao} and \emph{et al.}, ``Downlink {MIMO-NOMA} for ultra-reliable
  low-latency communications,'' \emph{IEEE J. Sel. Areas Commun.}, vol.~37,
  no.~4, pp. 780--794, Apr. 2019.

\bibitem{8301007}
M.~{Zeng}, A.~{Yadav}, O.~A. {Dobre}, and H.~V. {Poor}, ``Energy-efficient
  power allocation for {MIMO-NOMA} with multiple users in a cluster,''
  \emph{IEEE Access}, vol.~6, pp. 5170--5181, Feb. 2018.

\bibitem{8603758}
M.~{Zeng}, W.~{Hao}, O.~A. {Dobre}, and H.~V. {Poor}, ``Energy-efficient power
  allocation in uplink mmwave massive {MIMO} with {NOMA},'' \emph{IEEE Trans.
  Veh. Technol.}, vol.~68, no.~3, pp. 3000--3004, Mar. 2019.

\bibitem{7031971}
E.~{Bjornson}, L.~{Sanguinetti}, J.~{Hoydis}, and M.~{Debbah}, ``Optimal design
  of energy-efficient multi-user {MIMO} systems: Is massive {MIMO} the
  answer?'' \emph{IEEE Trans. Wireless Commun.}, vol.~14, no.~6, pp.
  3059--3075, June 2015.

\bibitem{8474292}
Z.~{Chang} and \emph{et al.}, ``Distributed resource allocation for energy
  efficiency in {OFDMA} multicell networks with wireless power transfer,''
  \emph{IEEE J. Sel. Areas Commun.}, vol.~37, no.~2, pp. 345--356, Feb. 2019.

\bibitem{8058662}
R.~{Hadani} and \emph{et al.}, ``Orthogonal time frequency space ({OTFS})
  modulation for millimeter-wave communications systems,'' in
  \emph{\textit{Proc.} 2017 IEEE MTT-S International Microwave Symposium
  (IMS)}, June 2017, pp. 681--683.

\bibitem{7009979}
G.~{Yang}, C.~K. {Ho}, R.~{Zhang}, and Y.~L. {Guan}, ``Throughput optimization
  for massive {MIMO} systems powered by wireless energy transfer,'' \emph{IEEE
  J. Sel. Areas Commun.}, vol.~33, no.~8, pp. 1640--1650, Aug. 2015.

\bibitem{8558585}
F.~{Yang}, W.~{Xu}, Z.~{Zhang}, L.~{Guo}, and J.~{Lin}, ``Energy efficiency
  maximization for relay-assisted {WPCN}: Joint time duration and power
  allocation,'' \emph{IEEE Access}, vol.~6, pp. 78\,297--78\,307, Dec. 2018.

\bibitem{7942090}
X.~{Wang} and C.~{Zhai}, ``Simultaneous wireless information and power transfer
  for downlink multi-user massive antenna-array systems,'' \emph{IEEE Trans.
  Commun.}, vol.~65, no.~9, pp. 4039--4048, Sep. 2017.

\bibitem{6884177}
C.~{Xiong}, L.~{Lu}, and G.~Y. {Li}, ``Energy efficiency tradeoff in downlink
  and uplink {TDD} {OFDMA} with simultaneous wireless information and power
  transfer,'' in \emph{\textit{Proc.} 2014 IEEE International Conf.
  Commun.(ICC)}, June 2014, pp. 5383--5388.

\bibitem{7273963}
Z.~{Ding}, P.~{Fan}, and H.~V. {Poor}, ``Impact of user pairing on {5G}
  nonorthogonal multiple-access downlink transmissions,'' \emph{IEEE Trans.
  Veh. Technol.}, vol.~65, no.~8, pp. 6010--6023, Aug. 2016.

\bibitem{7982794}
Y.~{Liu}, Z.~{Qin}, M.~{Elkashlan}, A.~{Nallanathan}, and J.~A. {McCann},
  ``Non-orthogonal multiple access in large-scale heterogeneous networks,''
  \emph{IEEE J. Sel. Areas Commun.}, vol.~35, no.~12, pp. 2667--2680, Dec.
  2017.

\bibitem{7488207}
Y.~{Zhang}, H.~{Wang}, T.~{Zheng}, and Q.~{Yang}, ``Energy-efficient
  transmission design in non-orthogonal multiple access,'' \emph{IEEE Trans.
  Veh. Technol.}, vol.~66, no.~3, pp. 2852--2857, Mar. 2017.

\bibitem{8413117}
Z.~{Chang}, L.~{Lei}, H.~{Zhang}, T.~{Ristaniemi}, S.~{Chatzinotas},
  B.~{Ottersten}, and Z.~{Han}, ``Energy-efficient and secure resource
  allocation for multiple-antenna {NOMA} with wireless power transfer,''
  \emph{IEEE Trans. Green Commun. Networking}, vol.~2, no.~4, pp. 1059--1071,
  Dec. 2018.

\bibitem{7974749}
B.~{Wang}, L.~{Dai}, Z.~{Wang}, and N.~{Ge}, ``Spectrum and energy-efficient
  beamspace {MIMO-NOMA} for millimeter-wave communications using lens antenna
  array,'' \emph{IEEE J. Sel. Areas Commun.}, vol.~35, no.~10, pp. 2370--2382,
  Oct. 2017.

\bibitem{8637821}
R.~{Jia}, X.~{Chen}, C.~{Zhong}, D.~W.~K. {Ng}, H.~{Lin}, and Z.~{Zhang},
  ``Design of non-orthogonal beamspace multiple access for cellular
  internet-of-things,'' \emph{IEEE J. Sel. Topics Signal Process.}, vol.~13,
  no.~3, pp. 538--552, June 2019.

\bibitem{6623072}
X.~{Chen}, X.~{Wang}, and X.~{Chen}, ``Energy-efficient optimization for
  wireless information and power transfer in large-scale {MIMO} systems
  employing energy beamforming,'' \emph{IEEE Wireless Commun. Lett.}, vol.~2,
  no.~6, pp. 667--670, Dec. 2013.

\bibitem{8119827}
Y.~{Yu}, H.~{Chen}, Y.~{Li}, Z.~{Ding}, L.~{Song}, and B.~{Vucetic}, ``Antenna
  selection for {MIMO} nonorthogonal multiple access systems,'' \emph{IEEE
  Trans. Veh. Technol.}, vol.~67, no.~4, pp. 3158--3171, Apr. 2018.

\bibitem{7536954}
A.~P. {Shrestha}, {Tao Han}, {Zhiquan Bai}, and K.~S. {Kwak}, ``Performance of
  transmit antenna selection in non-orthogonal multiple access for {5G}
  systems,'' in \emph{Proc. 8th Int. Conf. Ubiquitous Future Netw. (ICUFN)},
  July 2016, pp. 1031--1034.

\bibitem{8350092}
T.~A. {Zewde} and M.~C. {Gursoy}, ``{NOMA}-based energy-efficient wireless
  powered communications,'' \emph{IEEE Trans. Green Commun. Netw.}, vol.~2,
  no.~3, pp. 679--692, Sep. 2018.

\bibitem{7934461}
Z.~{Wei}, D.~W.~K. {Ng}, and J.~{Yuan}, ``Optimal resource allocation for
  power-efficient {MC-NOMA} with imperfect channel state information,''
  \emph{IEEE Trans. Commun.}, vol.~65, no.~9, pp. 3944--3961, Sep. 2017.

\bibitem{8119791}
F.~{Fang}, H.~{Zhang}, J.~{Cheng}, and V.~C.~M. {Leung}, ``Joint user
  scheduling and power allocation optimization for energy-efficient {NOMA}
  systems with imperfect {CSI},'' \emph{IEEE J. Sel. Areas Commun.}, vol.~35,
  no.~12, pp. 2874--2885, Dec. 2017.

\bibitem{8680645}
F.~{Fang}, Z.~{Ding}, W.~{Liang}, and H.~{Zhang}, ``Optimal energy efficient
  power allocation with user fairness for uplink {MC-NOMA} systems,''
  \emph{IEEE Wireless Commun. Lett.}, vol.~8, no.~4, pp. 1133--1136, Aug. 2019.

\bibitem{8186925}
S.~{Boyd}, N.~{Parikh}, E.~{Chu}, B.~{Peleato}, and J.~{Eckstein},
  ``Distributed optimization and statistical learning via the alternating
  direction method of multipliers,'' \emph{Found. Trends Mach. Learn.}, vol.~3,
  no.~1, pp. 1--122, Jan. 2011.

\bibitem{add202101Wei2012}
{E.{Wei} and A. E. {Ozdaglar}}, ``{Distributed alternating direction method of
  multipliers},'' in \emph{{Proc. IEEE Conf. Decis. Control}}, {Dec.} {2012},
  pp. {5445--5450}.

\bibitem{8632657}
M.~{Zeng}, A.~{Yadav}, O.~A. {Dobre}, and H.~V. {Poor}, ``Energy-efficient
  joint user-{RB} association and power allocation for uplink hybrid
  {NOMA-OMA},'' \emph{IEEE Internet Things J.}, vol.~6, no.~3, pp. 5119--5131,
  June 2019.

\bibitem{7070752}
K.~{Lee} and J.~{Hong}, ``Energy-efficient resource allocation for simultaneous
  information and energy transfer with imperfect channel estimation,''
  \emph{IEEE Trans. Veh. Technol.}, vol.~65, no.~4, pp. 2775--2780, Apr. 2016.

\bibitem{8771371}
{P. {Li} and Z. {Ding} and K. {Feng} }, ``{Enhanced receiver based on {FEC}
  code constraints for uplink {NOMA} with imperfect {CSI}},'' \emph{{IEEE
  Trans. Wireless Commun.}}, vol.~{18}, no.~{10}, pp. {4790--4802}, {{Oct.}},
  {2019}.

\bibitem{6725592}
H.~{Li}, L.~{Song}, and M.~{Debbah}, ``Energy efficiency of large-scale
  multiple antenna systems with transmit antenna selection,'' \emph{IEEE Trans.
  Commun.}, vol.~62, no.~2, pp. 638--647, Feb. 2014.

\bibitem{1321221}
S.{ Cui}, A.~J. {Goldsmith}, and A.~{Bahai}, ``Energy-efficiency of {MIMO} and
  cooperative {MIMO} techniques in sensor networks,'' \emph{IEEE J. Sel. Areas
  Commun.}, vol.~22, no.~6, pp. 1089--1098, Aug. 2004.

\bibitem{WD}
W.{ Dinkelbach}, ``On nonlinear fractional programming,'' \emph{Manage. Sci},
  vol.~13, pp. 492--498, Mar. 1967.

\bibitem{1658226}
W.~{ Yu} and R.~{Lui}, ``Dual methods for nonconvex spectrum optimization of
  multicarrier systems,'' \emph{IEEE Trans. Commun.}, vol.~54, no.~7, pp.
  1310--1322, July 2006.

\bibitem{6364677}
D.~W.~K. {Ng} and R.~{Schober}, ``Energy-efficient resource allocation in
  {OFDMA} systems with large numbers of base station antennas,'' in \emph{2012
  IEEE International Conf. Commun. (ICC)}, Ottawa, Canada, Jun. 2012, pp.
  5916--5920.

\bibitem{add20201}
T.~{Tanino} and Y.~{Sawaragi}, ``Duality theory in multiobjective
  programming,'' \emph{J. Optimization Theory and Applications}, vol.~27,
  no.~4, pp. 509--529, 1979.

\bibitem{JEckstein}
J.~{ Eckstein}, ``Augmented lagrangian and alternating direction methods for
  convex optimization: A tutorial and some illustrative computational
  results,'' Rutgers Univ., New Brunswick, NJ, USA, RUTCOR Res. Rep., vol. 32,
  2012.

\bibitem{6644242}
M.~{Leinonen} and \emph{et al.}, ``Distributed joint resource and routing
  optimization in wireless sensor networks via alternating direction method of
  multipliers,'' \emph{IEEE Trans. Wireless Commun.}, vol.~12, no.~11, pp.
  5454--5467, Nov. 2013.

\bibitem{Boyd2004Convex}
S.~{Boyd} and L.~{Vandenberghe}, \emph{Convex Optimization}.\hskip 1em plus
  0.5em minus 0.4em\relax Cambridge, U.K.: Cambridge Univ. Press, 2004.

\bibitem{7968315}
{K. {Nguyen} and Q. {Vu} and M. {Juntti} and L. {Tran}}, ``{Distributed
  solutions for energy efficiency fairness in multicell {MISO} downlink},''
  \emph{{IEEE Trans. Wireless Commun.}}, vol.~{16}, no.~{9}, pp. {6232--6247},
  {Sep.} {2017}.

\bibitem{add2020410}
{ L. {Xiao} and S. {Boyd}}, ``{Fast linear iterations for distributed
  averaging},'' \emph{{Syst. Control Lett.}}, vol.~{53}, no.~{1}, pp. {65--78},
  {2004}.

\bibitem{7335602}
{Q. {Vu} and L. {Tran} and R. {Farrell} and E. {Hong}}, ``{Energy-efficient
  zero-forcing precoding design for small-cell networks},'' \emph{{IEEE Trans.
  Commun.}}, vol.~{64}, no.~{2}, pp. {790--804}, {Feb.} {2016}.

\bibitem{7332956}
Q.~{Wu}, M.~{Tao}, D.~W. {Kwan Ng}, W.~{Chen}, and R.~{Schober},
  ``Energy-efficient resource allocation for wireless powered communication
  networks,'' \emph{IEEE Trans. Wireless Commun.}, vol.~15, no.~3, pp.
  2312--2327, Mar. 2016.

\bibitem{7582543}
P.~D. {Diamantoulakis}, K.~N. {Pappi}, Z.~{Ding}, and G.~K. {Karagiannidis},
  ``Wireless-powered communications with non-orthogonal multiple access,''
  \emph{IEEE Trans. Wireless Commun.}, vol.~15, no.~12, pp. 8422--8436, Dec.
  2016.

\bibitem{Stigler}
S.~M. Stigler, ``The asymptotic distribution of the trimmed mean,''
  \emph{Annals of Statistics}, vol.~1, no.~3, pp. 472--477, May. 1973.

\bibitem{5766166}
P.~{Hesami} and J.~N. {Laneman}, ``Limiting behavior of receive antennae
  selection,'' in \emph{2011 45th Annual Conference on Information Sciences and
  Systems}, Baltimore, MD, USA, Mar. 2011, pp. 1--6.

\bibitem{7906591}
Z.~{Chang}, Z.~{Wang}, X.~{Guo}, Z.~{Han}, and T.~{Ristaniemi},
  ``Energy-efficient resource allocation for wireless powered massive {MIMO}
  system with imperfect {CSI},'' \emph{IEEE Trans. Green Commun. Netw.},
  vol.~1, no.~2, pp. 121--130, Apr. 2017.
\end{thebibliography}
\end{document}